\documentclass[journal]{IEEEtran}
\usepackage{cite}
\usepackage{enumerate}
\usepackage{amsmath}
\usepackage{amsfonts}
\usepackage{mathrsfs}
\usepackage{amssymb}
\usepackage{cases}
\usepackage{stmaryrd}
\usepackage{mathrsfs}
\usepackage{bbm}
\usepackage{algorithm}
\usepackage{algorithmic} 
\usepackage{multirow}
\usepackage{graphicx}
\usepackage{float}
\usepackage{subfig}
\usepackage{bm}
\usepackage{color}
\usepackage{setspace} 
\usepackage[acronym]{glossaries}
\usepackage{booktabs}
\usepackage[colorlinks]{hyperref}
\usepackage{threeparttable}
\usepackage{multirow}
\usepackage{setspace}
\usepackage{makecell}
\usepackage[table]{xcolor}
\usepackage[font=footnotesize]{caption} 
\newcommand{\INPUT}{\item[\textbf{Input:}]}
\newcommand{\OUTPUT}{\item[\textbf{Output:}]}
\usepackage{pdflscape} 
\newtheorem{lemma}{Lemma}
\newtheorem{theorem}{Theorem}

\newacronym{BS}{BS}{base station}
\newacronym{MIMO}{MIMO}{multiple-input multiple-output}
\newacronym{AI}{AI}{artificial intelligence}
\newacronym{UPA}{UPA}{uniform planar array}
\newacronym{ULA}{ULA}{uniform linear array}
\newacronym{CSCG}{CSCG}{circularly symmetric complex Gaussian}
\newacronym{THz}{THz}{terahertz}
\newacronym{CSI}{CSI}{channel state information}
\newacronym{OFDM}{OFDM}{orthogonal frequency-division multiplexing}
\newacronym{ISAC}{ISAC}{integrated sensing and communication}
\newacronym{RF}{RF}{radio frequency}
\newacronym{LOS}{LoS}{ line-of-sight}
\newacronym{NLoS}{NLoS}{non-line-of-sight}
\newacronym{AWGN}{AWGN}{additive white Gaussian noise}
\newacronym{PS}{PS}{phase shift}
\newacronym{TTD}{TTD}{true time delayer}
\newacronym{RCS}{RCS}{radar cross-section}
\newacronym{FIM}{FIM}{Fisher information matrix}
\newacronym{GNN}{GNN}{graph neural network}
\newacronym{ELAA}{ELAA}{extremely large-scale antenna array}
\newacronym{DoF}{DoF}{degrees-of-freedom}
\newacronym{CRB}{CRB}{Cram\'{e}r-Rao bound}
\newacronym{PCRB}{PCRB}{posterior Cram\'er--Rao bound}
\definecolor{kwan}{rgb}{0,0,1}

\newcommand{\tabincell}[2]{\begin{tabular}{@{}#1@{}}#2\end{tabular}}

\begin{document}
	\title{Secure Cooperative THz ISAC via Mamba Empowered Graph Neural Network Precoding
		\author{Chao Wang,~\IEEEmembership{Senior Member,~IEEE,}
				Zan Li,~\IEEEmembership{Fellow,~IEEE,}
				Xiangnan Zhou,
				Haibin Zhang,
				Hao Xu,
				Liang Jin, \\
			and	Derrick Wing Kwan Ng,~\IEEEmembership{Fellow, IEEE}
			
			\vspace{-.36in}
			\thanks{Part of this paper was presented at the 2025 IEEE Global Communications Conference (GLOBECOM'25) \cite{GLOBECOM25}}
			\thanks{The work is supported by National Natural Science Foundation of China under Grant 62371357, Funds for International Cooperation and Exchange of the National Natural Science Foundation of China under Grant 12411530120, Key
				Research and Development Program of Shaanxi(ProgramNo.2025GHYBXM048), National Science Fund for Distinguished Young Scholars of China under
				Grant 62121001, Major National Science and Technology Project of China
				(Grant No. 20252D1303100), National Natural Science Foundation of China
				under Grant U22A2001. (\textit{Corresponding authors: Chao Wang and Zan Li})}
	\thanks{C. Wang and Z. Li are  with the Integrated Service Networks Lab, Xidian University, Xi'an 710071, China (e-mail: drchaowang@126.com).}
\thanks{X. N. Zhou and H. B. Zhang are with the School of Cyber Engineering,
	Xidian University, Xi'an 710071, China (e-mail: liguo@stu.xidian.edu.cn)}
\thanks{H. Xu is with the National Mobile Communications Research Laboratory,
Southeast University, Nanjing 210096, China (e-mail: hao.xu@seu.edu.cn).}
\thanks{L. Jin is with Information Engineering University, Zhengzhou, 450001,
	and also with Peng Cheng Laboratory, Shenzhen 518066, China (e-mail:
	liangjin@263.net).}
\thanks{D. W. K. Ng is with the School of Electrical Engineering and Telecommunications, University of New South Wales, Sydney, 2052, Australia (w.k.ng@unsw.edu.au).}

		}

	}
	
	\maketitle
	
	\begin{abstract}
The terahertz (THz) band offers abundant  spectrum resources for high-throughput communication and ultra high-precision localization. This paper investigates secure communication in cooperative THz orthogonal frequency-division multiplexing (OFDM) bistatic integrated sensing and communications (ISAC) systems, where multiple base stations (BSs) equipped with extremely large-scale antenna arrays (ELAAs) collaboratively serve downlink users while concurrently locating multiple targets. Malicious targets are assumed to act as potential eavesdroppers attempting to intercept confidential information intended for legitimate users.
To mitigate these threats, we formulate a joint optimization problem for analog beamforming, digital precoding, true-time delayers (TTDs), and sensing signal covariance matrix design. The objective is to maximize the minimum secrecy rate subject to Cramér-Rao bound (CRB) constraints that ensure localization accuracy. This problem is highly challenging due to the non-convex CRB constraint, strongly coupled variables, high computational complexity from ELAA, and near-field channel modeling.
To address these challenges, we propose a novel data-driven framework that integrates \glspl{GNN} with the Mamba architecture. Our proposed framework first encodes the interactions among users, targets, and BSs into a heterogeneous graph and then employs message passing to optimize vertex features.  The Mamba blocks further enhance this process through their selection mechanism and state space modeling capabilities, enabling dynamic and context-aware optimization of beamforming, TTD configurations, and sensing parameters. {Numerical simulations validate that the proposed method outperforms both conventional and learning-based baselines, while offering high computational efficiency and strong generalization across different network conditions.}
	\end{abstract}
	
	
	\IEEEpeerreviewmaketitle
	
	\vspace{-2mm}
	\section{Introduction}

Anticipated to emerge around 2030, six-generation (6G) wireless networks will build upon fifth-generation (5G) technology to deliver significant performance enhancements \cite{6G_Vision}. These include peak data rates exceeding 100 Gb/s (an order of magnitude faster than 5G), a tripling of spectral efficiency, and support for connection densities of 10 million devices per square kilometer. To realize these capabilities, 6G must utilize the higher-frequency \gls{THz} spectrum to overcome the inherent spectral scarcity of sub-6 GHz bands, though this introduces significant propagation challenges.

\gls{THz} communications encounter significant challenges, including severe propagation losses due to atmospheric absorption and free-space path loss, which significantly restricts transmission distances. Additionally, the ultra-broad bandwidth of \gls{THz} signals suffer from pronounced multipath fading and strong frequency selectivity. To overcome these issues, \gls{ELAA} and \gls{OFDM} have emerged as key enabling technologies. Indeed, recent research has focused on THz signal processing techniques, including \gls{ELAA} channel modeling \cite{ChannelModeling}, beamforming optimization \cite{ELAA_Beamforming}, mitigation of beam split across \gls{OFDM} subcarriers \cite{Delay_Phase_TTD}, and \gls{ISAC} implementations \cite{THZ_ISAC_Survey,NearField_ISAC,10477314}.
Notably, as antenna dimensions scale to \gls{ELAA} configurations, the increased Fraunhofer distance places both downlink users and radar targets in the radiating near-field (Fresnel) region, where spherical-wave propagation models  replace conventional plane-wave approximations \cite{ChannelModeling}.
Furthermore, the wide bandwidth of THz  leads to significant beam-split effects, where different frequency components of a signal are steered in divergent directions. To mitigate this, \cite{Delay_Phase_TTD} proposed delay-phase precoding, which employs both \gls{TTD} and phase shifters to effectively compensate for beam split.

Among these enabling techniques, THz \gls{ISAC} stands out as a transformative paradigm that seamlessly merges high-speed communication with ultra-precise sensing, thus driving innovations for embracing 6G \cite{THZ_ISAC_Survey,b6}. 
Such dual functionality unlocks critical applications ranging from low-altitude economy operations to intelligent transportation systems \cite{ISAC_Survey}. 
Research on \gls{ISAC} has grown exponentially recently, with most studies focusing on \gls{MIMO} techniques aimed at optimizing spatial signal patterns for enhanced \gls{ISAC} performance across diverse scenarios \cite{THz-ISAC-PHY}.

Despite various efforts have been devoted, the practical deployment of \gls{ISAC} faces significant challenges, particularly inter-cell interference and the severe two-hop path loss in target sensing, which limit system performance~\cite{Network_ISAC2}. 
To effectively address these issues, the authors of~\cite{Network_ISAC2} proposed novel network-\gls{ISAC} architectures that rely on multi-\gls{BS} cooperation to enhance the performance, and investigated the allocation of antennas among \glspl{BS} to optimize cooperative \gls{ISAC} performance. 
{
However, beyond the performance benefits brought by cooperative multi-\gls{BS} architectures, \gls{ISAC} also raises important secrecy concerns. These concerns fundamentally stem from the reuse of communication signals for sensing~\cite{10443616,ISAC_security,Secrecy_ISAC,Covert_ISAC}, through which the communication component may be received outside the intended users and thus cause information leakage~\cite{ISAC_security}. In cooperative multi-\gls{BS} settings, this risk can become even more severe due to the wider signal coverage. More importantly, sensing targets are often located in strongly illuminated regions and may, if malicious, act as potential eavesdroppers, thereby giving rise to a distinctive security challenge in \gls{ISAC} and necessitating a tradeoff between sensing performance and secure communication.}
In response to these security challenges, a growing body of work has investigated secure \gls{ISAC} design from different perspectives \cite{11322745}. For example, \cite{ISAC_security} proposed a symbol-level precoding scheme enabled by simultaneously transmitting and reflecting reconfigurable intelligent surfaces (STAR-RIS) to secure information transmission and target sensing. \cite{Covert_ISAC} studied covert \gls{ISAC} design for protecting information transmission against malicious detection. \cite{NearFieldISACSecletter} proposed a near-field secure \gls{ISAC} system architecture that performs target detection and tracking while achieving secure communications. In addition, \cite{PLS_ISAC} reviewed recent research on \gls{ISAC} system design for protecting confidential information. 
{It is worth noting that \gls{THz} communication's extreme frequency enables massive antenna arrays to produce ultra-narrow, pencil-like beams, while its high atmospheric absorption and susceptibility to blockage by walls or people naturally confine signals to short-range, line-of-sight paths, which physically prevent remote eavesdropping. These inherent physical-layer characteristics make \gls{THz} uniquely suited for secure, high-speed wireless links \cite{10227529}. Accordingly, a growing number of studies have explored the potential of \gls{THz} in securing \gls{ISAC} \cite{GLOBECOM25,10443616}. First, \gls{THz} frequencies enhance the sensing capability of ISAC systems by enabling ultra-high resolution, fine Doppler sensitivity, and precise target localization, thanks to their extremely short wavelength and large available bandwidth. Second, \gls{THz} inherently protects information from wiretapping through its physical properties.} 
{Despite extensive research on \gls{ISAC} security, the authors find that a comprehensive framework jointly encompassing the distinctive features of \gls{THz}-\gls{ISAC}, such as near-field effects, beam-splitting, information security, and cooperative design, remains largely unexplored (as summarized in Table I). Furthermore, traditional optimization approaches for wireless systems are ill-suited for \gls{THz}-\gls{ISAC}. For example, the work of \cite{10443616} proposed optimization-based approaches for \gls{THz}-\gls{ISAC}, whose computational complexity becomes prohibitive for massive antenna arrays. More seriously, the computational burden becomes unbearable for cooperative \gls{THz}-\gls{ISAC}.}
Consequently, \gls{ISAC} transforms wireless networks into versatile multi-taskers capable of handling complex scenarios that lie beyond the reach of conventional network designs and demand more intelligent solutions. 
	\begin{figure}
	\vspace{-4mm}
	\centering
	\includegraphics[width=3.6in]{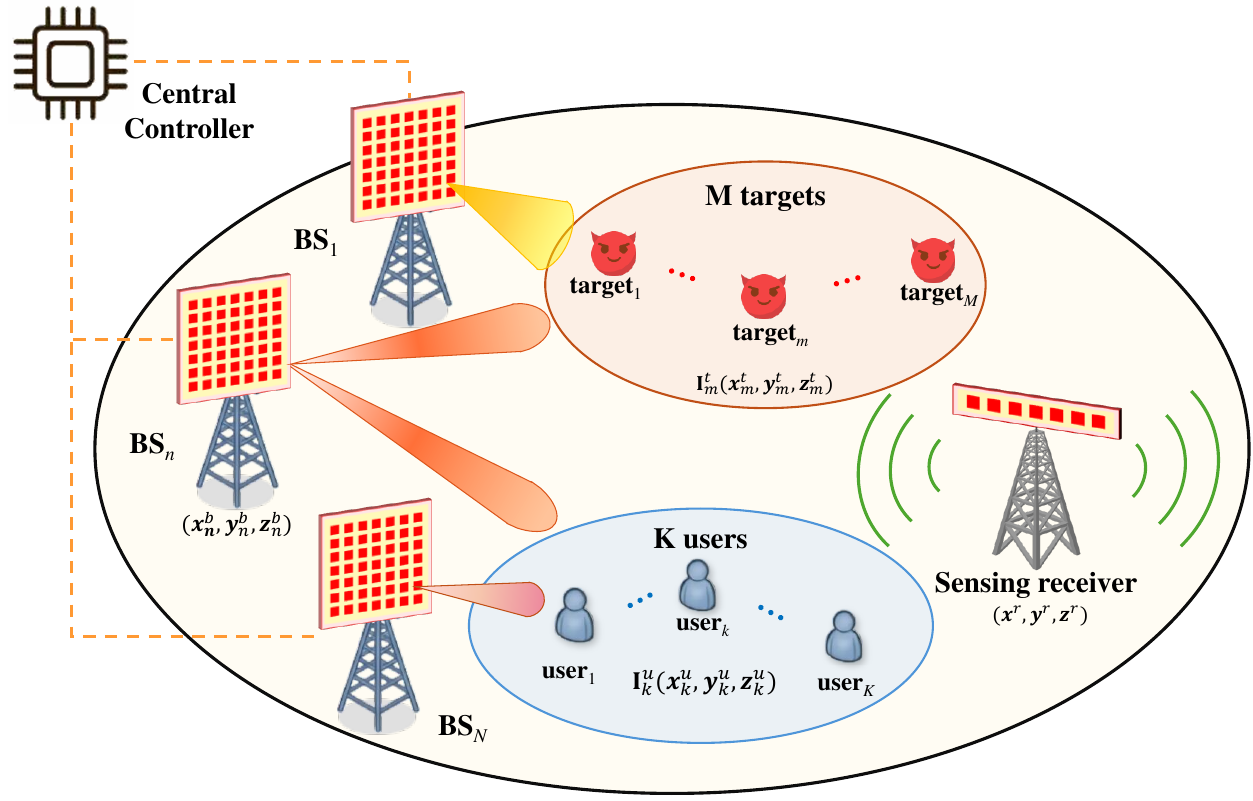}
	\caption{The proposed ISAC system model, with $N$ BSs, a sensing receiver, $M$ targets, and $K$ users.}\label{system}
\end{figure}


	{ In contrast to 5G, where \gls{AI} served as an auxiliary enhancement \cite{10599127,11214530}, 6G is envisioned to possess native intelligence as a core functionality and defining attribute~\cite{AI-native-6G}. 
	Accordingly, recent research has increasingly focused on integrating \gls{AI} with \gls{ISAC} systems~\cite{Learning-ISAC1,Learning-ISAC2,Learning-ISAC3}. 
	Notably, \cite{Learning-ISAC1} developed two deep learning algorithms to address symbol-level precoding challenges in \gls{ISAC} waveform design, while \cite{Learning-ISAC2} proposed a low-complexity deep unfolding approach for \gls{ISAC} transceiver design in cluttered environments. 
	By replacing iterative online optimization with a single forward inference process, these  methods can substantially reduce the computational burden compared with conventional optimization approaches. 
	However, existing research works relied on deep learning to optimize \gls{ISAC} precoders, which often suffer from limited generalization in dynamic network topologies due to their fixed input/output dimensions and topology sensitivity. 
	In this context, \glspl{GNN} emerge as a promising solution, having demonstrated unique effectiveness in graph-related applications while offering the adaptability needed for heterogeneous network architectures \cite{GNN_Wireless}.  
	Specifically, through structured graph representations, \glspl{GNN} explicitly capture the underlying topological relationships within wireless networks and leverage iterative message passing to learn expressive representations that reflect both local interference patterns and network-level spatial structures~\cite{permutation2}. Furthermore, the inherent permutation invariance and inductive bias of \glspl{GNN} endow the models with superior scalability, allowing them to generalize well across varying spatial distributions ~\cite{permutation1,attention}.
	Building upon these inherent advantages, recent works have  applied \glspl{GNN} to the \gls{ISAC}.
	For example, \cite{GNN_ISAC1} employed a \gls{GNN} with an implicit projection framework to optimize radar transmit beampatterns under communication constraints, thereby improving scalability. 
	}

\begin{table}[!t]\label{Table}
	\caption{{Overview of existing literature on \gls{ISAC}}}
	\small
	\centering
	\begin{tabular}{c|c|c|c|c|}
		\hline\rowcolor{lightgray}
	\textbf{Ref} & \makecell{Near-Field \\ Propagation} & \makecell{Information \\ Security}   &\makecell{Cooperative \\ Design} & \makecell{Beam \\ Splitting} \\ 
			\hline
	\cite{NearField_ISAC} &$\checkmark$ && &  \\ 
		\hline
	\cite{b6}& $\checkmark$&& &  \\ 
				\hline
	\cite{THz-ISAC-PHY}&$\checkmark$ && $\checkmark$&  \\ 
	\hline
		\cite{Network_ISAC2}& && $\checkmark$&  \\ 	
		\hline
				\cite{10443616}& &$\checkmark$& $\checkmark$ &  \\ 
		\hline 
	  	  \cite{ISAC_security}& &$\checkmark$& $\checkmark$&  \\ 
	   \hline
\cite{Secrecy_ISAC}& &$\checkmark$& $\checkmark$&  \\ 	
\hline 
\cite{Covert_ISAC}& &$\checkmark$& &  \\ 
		\hline 
		\cite{NearFieldISACSecletter}&$\checkmark$ &$\checkmark$& &  \\ 
		\hline 
		\cite{Learning-ISAC3}& &$\checkmark$& &  \\ 
          \hline
         \textbf{Prop.}& $\checkmark$&$\checkmark$&$\checkmark$ &  $\checkmark$\\ 
         \hline
	\end{tabular}
\end{table}

	{Against this background, this work investigates a cooperative \gls{THz} \gls{ISAC} scenario in which multiple \glspl{BS}, equipped with \gls{ELAA}, collaboratively serve several downlink users while simultaneously locating multiple targets. Some of these targets may act as malicious eavesdroppers that attempt to intercept confidential information intended for legitimate users.
	Although secure \gls{ISAC} is an increasingly studied topic, few works have considered its cooperative variant due to the higher complexity that cooperation entails. These challenges are further compounded by the use of \gls{ELAA} and the ultra-wide \gls{THz} band, which introduce additional complexity into the cooperative secure design. Consequently, the high computational complexity of traditional optimization-based approaches necessitates more intelligent and lower complexity methods. However, few studies have proposed learning-aided optimization approaches with low-complexity for this problem.
	Therefore, to address this gap, we propose a Mamba-empowered \gls{GNN} framework for scalable joint optimization in bistatic \gls{ISAC} systems. In contrast to the conference version \cite{GLOBECOM25}, this work employs \gls{TTD} to mitigate the beam split effect inherent in wideband systems. Moreover, specifically to enhance generalization ability, we replace the conventional deep neural network framework used in the conference version with the proposed Mamba-empowered \gls{GNN} architecture.} The main contributions of this work are summarized as follows:

	1) We develop an efficient cooperative \gls{THz} \gls{ISAC} framework that addresses three fundamental challenges: (i) beam split effects caused by the ultra-wide bandwidth of the \gls{THz} band, (ii) near-field characteristic due to \gls{ELAA} configurations, and (iii) scalability requirements in different wireless networks. Exploiting \glspl{GNN}, we propose a data-driven learning framework that jointly optimizes analog beamforming, digital precoding, radar covariance matrix, and \gls{TTD} networks, for maximizing the minimum secrecy rate while satisfying a tolerable \gls{CRB} constraint.
	
	2) Our learning architecture effectively leverages \glspl{GNN} for their ability to capture complex relationships and adapt to dynamic network conditions. The framework processes graph-structured data, rendering it ideal for modeling interactions between  targets, communication users, and \glspl{BS}, as well as handling topology variations. To further reduce computational complexity, we incorporate a Mamba module to exploit features obtained through message passing for designing analog beamforming and \gls{TTD} networks.
	
	3) {{Simulation results show that the proposed Mamba-empowered \gls{GNN} approach outperforms the conventional Alt-Min scheme~\cite{Alt-Min} and the considered learning-based benchmarks~\cite{murshed2023cnn}. Compared with the Alt-Min scheme without sensing constraints, the proposed method achieves higher secrecy rates with lower computational cost while satisfying the sensing requirements. Moreover, it consistently surpasses the deep learning baselines and shows robust generalization across different network topologies.}}

	\begin{figure}[t]
	\centering
	\includegraphics[width=3.3in]{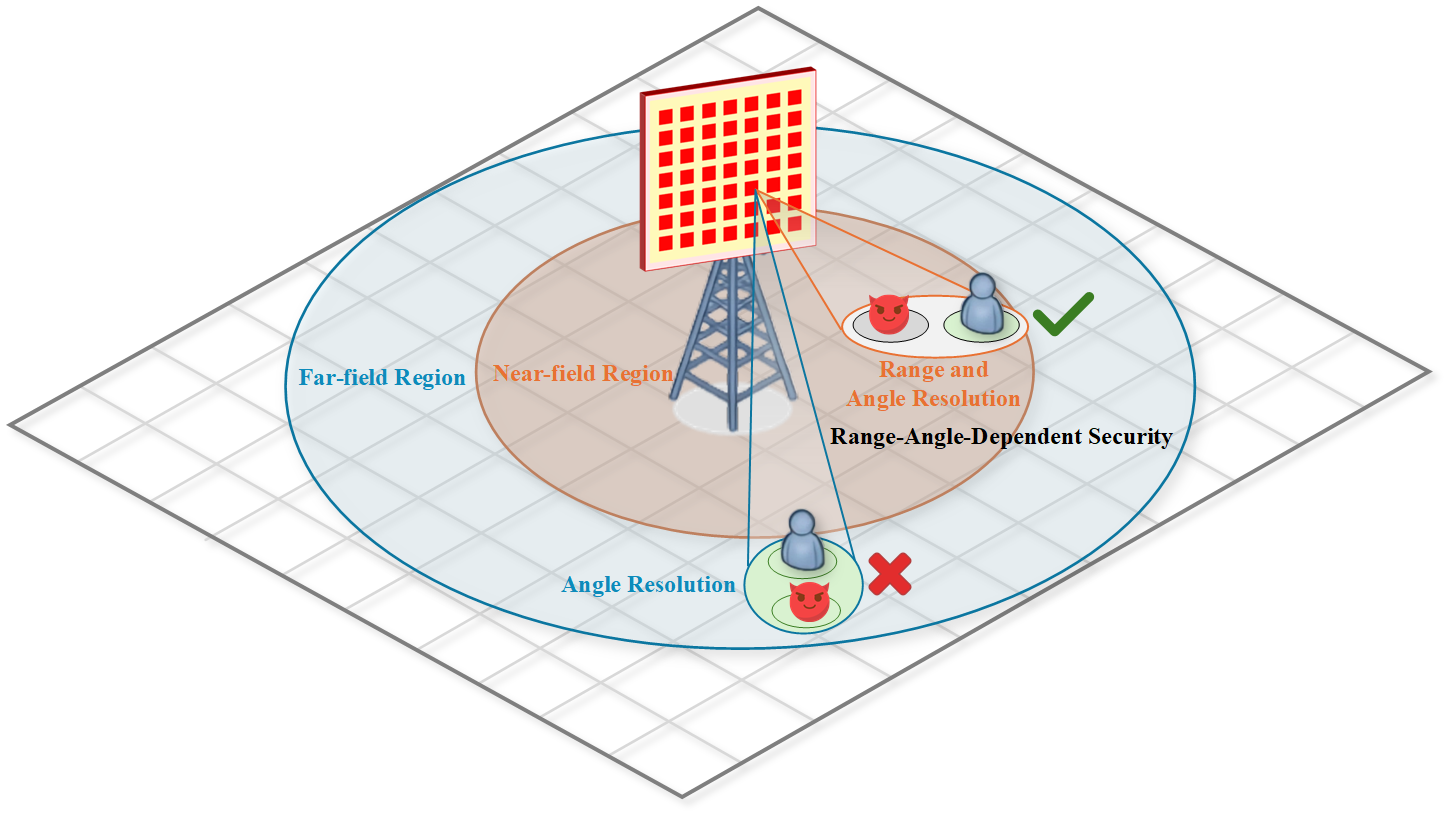}
	\caption{Conceptual illustration of the near-field effect: Near-field propagation enables joint angle-and-range resolution, whereas far-field propagation only provides angular resolution.}
	\label{fig:nearfield_effect}
	\end{figure}

	\textit{Notations:} Scalars, vectors, and matrices are indicated as lowercase letters, lowercase boldfaced letters, and uppercase boldfaced letters, respectively. $(\cdot)^\mathrm{T}$, $(\cdot)^\mathrm{H}$, and $(\cdot)^*$ denote the transpose, transpose conjugate, and conjugate operation, respectively. $\text{vec}(\mathbf{A})$ and $\text{tr}(\mathbf{A})$ denote the vectorization and trace of matirx $\mathbf{A}$, respectively. $\mathbb{C}^{N\times M}$ and $\mathbb{R}^{N\times M}$ denote the space of $N\times M$ complex matrices and real matrices. $\left|\cdot \right|$, $\Re(\cdot)$ and $\Im(\cdot)$ denote the absolute value, the real part, and the imaginary part of a complex entry, respectively. $\mathbf{I}_N$ indicates the $N \times N$ identity matrix. $\otimes$ and $\odot$ indicate Kronecker product and Hadamard product, respectively. $\mathbb{E}(\cdot)$ denotes the statistical expectation. $\|\cdot\|$is the Euclidean norm of a vector, and $\|\cdot\|_F$ denotes
	Frobenius-norm of the matrix. $\text{blkdiag}(\cdot)$ denotes the block-diagonal matrix 
	and $\text{diag}(\mathbf{a})$ denotes a diagonal matrix with diagonal elements $\mathbf{a}$. \( \mathcal{CN}(\mathbf{x}, \mathbf{Y}) \) denotes a circularly symmetric complex Gaussian (CSCG) random vector with mean vector $\mathbf{x}$ and covariance matrix $\mathbf{Y}$.
	
	\section{System Model and Problem Formulation}  
	As shown in Fig.~\ref{system}, we consider a bistatic \gls{THz} \gls{ISAC} system where $N$ \gls{BS}s, each equipped with an $N_\mathrm{t}$-antenna \gls{UPA}, cooperatively transmit information to $K$ single-antenna users while locating $M$ targets. Due to the \gls{ELAA} configuration ($N_\mathrm{t} \gg 1$), multiple users and targets lie in the near-field region of the \gls{BS}s~\cite{Alt-Min}. A receiving \gls{BS}, equipped with an $N_\mathrm{r}$-element \gls{ULA}, forms a bistatic pair with the $N$ transmitting \gls{BS}s. 
	{In the considered system, the $M$ sensing targets are modeled as users that require sensing services but are not scheduled for communication service in the current transmission interval.
	As these targets may intercept confidential information signals intended for the scheduled communication users, they are treated as potential eavesdroppers.
	Following common assumptions in secure ISAC studies, the channel information of both scheduled communication users and sensing targets is assumed to be available at the BSs for transmit design~\cite{channel1,channel2,channel3}.}
	
	{As illustrated in Fig.~\ref{fig:nearfield_effect}, near-field propagation in the considered system provides joint angle-and-range resolution. This property extends conventional angular-domain spatial multiplexing to angle--range-domain spatial multiplexing, enabling spatial isolation in both the angular and range domains and thereby helping mitigate potential eavesdropping threats from targets located in similar angular directions as scheduled users. Meanwhile, the angle--range focusing capability of near-field propagation reduces energy spreading around the target and improves localization accuracy.}

	The center of the $n$-th \gls{BS}'s \gls{UPA} is located at $\mathbf{l}_n^{\mathrm{b}} \triangleq (x_n^{\mathrm{b}}, y_n^{\mathrm{b}}, z_n^{\mathrm{b}})$, with its $n_t$-th antenna element at $\mathbf{l}_{n_t}^{\mathrm{b},n} \triangleq \left(x_{n_t}^{\mathrm{b},n}, y_{n_t}^{\mathrm{b},n}, z_{n_t}^{\mathrm{b},n}\right)$. Similarly, the receiving BS's \gls{ULA} has its center at $\mathbf{l}^{\mathrm{r}} \triangleq \left(x^{\mathrm{r}}, y^{\mathrm{r}}, z^{\mathrm{r}}\right)$, with its $n_r$-th antenna element at $\mathbf{l}_{n_r}^{\mathrm{r}} \triangleq \left(x_{n_r}^{\mathrm{r}}, y_{n_r}^{\mathrm{r}}, z_{n_r}^{\mathrm{r}}\right)$. The $m$-th target and the $k$-th user are located at $\mathbf{l}_{m}^{\mathrm{t}} \triangleq \left(x_{m}^{\mathrm{t}}, y_{m}^{\mathrm{t}}, z_{m}^{\mathrm{t}}\right)$ and $\mathbf{l}_{k}^{\mathrm{u}} \triangleq \left(x_{k}^{\mathrm{u}}, y_{k}^{\mathrm{u}}, z_{k}^{\mathrm{u}}\right)$, respectively.
	
	\begin{figure}
		\vspace{-4mm}
		\centering
		\includegraphics[width=3in]{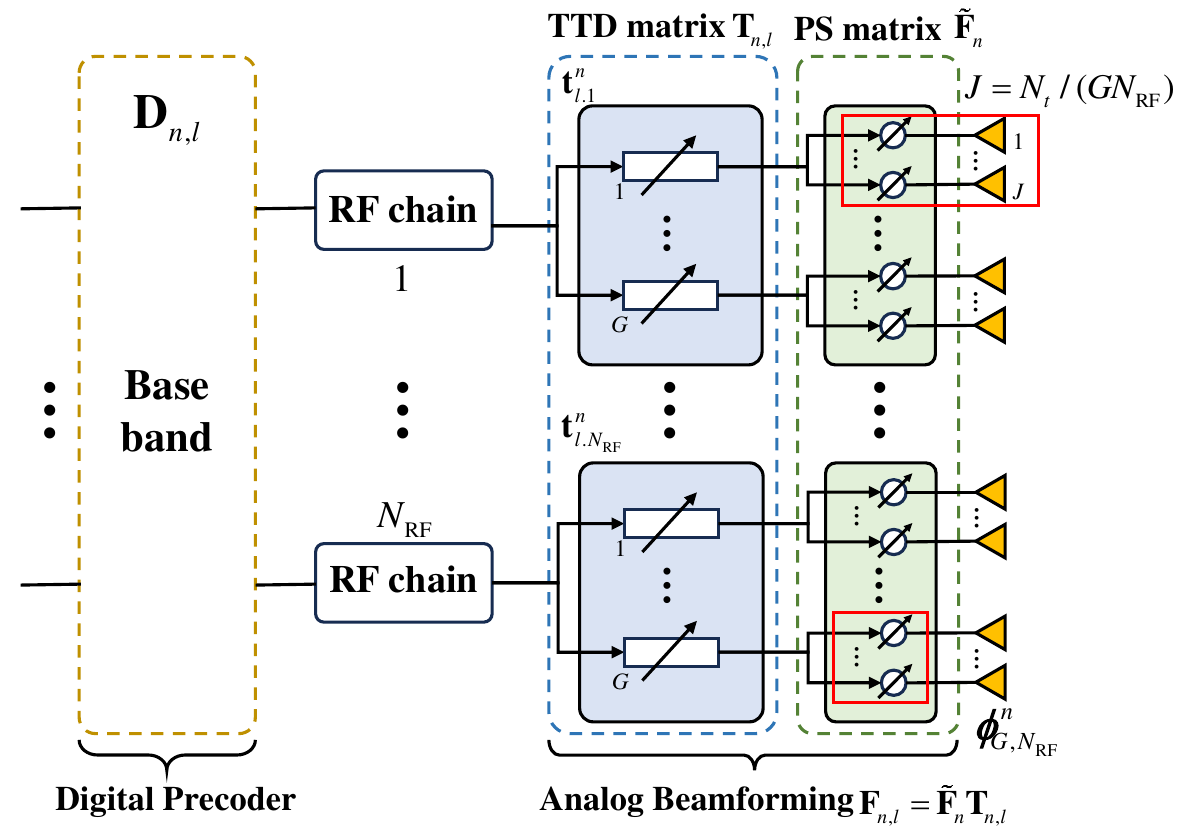}
		\caption{Hybrid architecture transmitter.}\label{Hybrid_Structure}
	\end{figure}		
	
	We adopt \gls{OFDM} with $L$ subcarriers and a partially connected hybrid precoding architecture, as shown in Fig.~\ref{Hybrid_Structure}. {In the considered system, all cooperative BSs share the same set of subcarriers, and multiple users are simultaneously served on each subcarrier through linear precoding.} The transmitted signal from the $n$-th \gls{BS} at the $l$-th subcarrier and time instant $t$ is given by:
	\begin{align}
		\bm{\chi}_{n,l}(t) = \mathbf{F}_{n,l} \left( \sum_{k=1}^{K} \mathbf{d}_{nk,l} c_{k,l}^{\mathrm{I}}(t) + \mathbf{c}_{n,l}^{\mathrm{S}}(t) \right),
	\end{align}
	where \( c_{k,l}^{\mathrm{I}}(t) \sim\mathcal{CN}(0,1) \) is the information-bearing signal for the $k$-th user and \(\mathbf{c}_{n,l}^{\mathrm{S}}(t) \in \mathbb{C}^{N_{\text{RF}} \times 1}\) is the radar waveform. The analog beamforming matrix is defined as \(\mathbf{F}_{n,l}= \widetilde{\mathbf{F}}_n \mathbf{T}_{n,l}\), where \(\widetilde{\mathbf{F}}_n \in \mathbb{C}^{N_t \times G N_{\mathrm{RF}}}\) is a frequency-independent phase-shifting (PS) matrix given by $
	\widetilde{\mathbf{F}}_n \triangleq \mathrm{blkdiag}(\boldsymbol{\phi}_{1,1}^n, \ldots, \boldsymbol{\phi}_{g,n_{\mathrm{RF}}}^n, \ldots, \boldsymbol{\phi}_{G, N_{\mathrm{RF}}}^n).$, \(\mathbf{T}_{n,l} \in \mathbb{C}^{G N_{\mathrm{RF}} \times N_{\mathrm{RF}}}\) represents the frequency-dependent phase shifts introduced by \gls{TTD} units, structured as $
	\mathbf{T}_{n,l} \triangleq \mathrm{blkdiag}(\mathbf{t}_{l,1}^n, \mathbf{t}_{l,2}^n, \ldots, \mathbf{t}_{l,N_{\mathrm{RF}}}^n)$, with $G$ denoting the number of TTDs per RF chain.
	
	The subvector \(\boldsymbol{\phi}_{g,n_{\mathrm{RF}}}^n \in \mathbb{C}^{J \times 1}\), where \(J = N_t / (G N_{\mathrm{RF}})\), contains the phase shifts applied to the \(J\) antennas connected to the \(g\)-th TTD of the \(n_{\mathrm{RF}}\)-th RF chain. Similarly, \(\mathbf{t}_{l,n_{\mathrm{RF}}}^n \in \mathbb{C}^{G \times 1}\) represents the frequency-dependent phase shifts of the \(G\) TTDs associated with the \(n_{\mathrm{RF}}\)-th RF chain, i.e., $\mathbf{t}_{l,n_{\mathrm{RF}}}^n = \exp\left( -j 2\pi f_l \boldsymbol{\tau}_{n_{\mathrm{RF}}}^n \right)$,
	where $\boldsymbol{\tau}_{n_{\mathrm{RF}}}^n = [\tau_{1,n_{\mathrm{RF}}}^n, \tau_{2,n_{\mathrm{RF}}}^n, \ldots, \tau_{G,n_{\mathrm{RF}}}^n]^{\mathrm{T}} \in \mathbb{R}^{G \times 1}$ represents the time delay vector and $f_l$ denotes the center frequency of the $l$th subcarrier. The frequency-dependent digital precoder \(\mathbf{d}_{nk,l} \in \mathbb{C}^{N_{\text{RF}} \times 1}\) at the \(n\)-th BS is designed to serve the \(k\)-th user on the \(l\)-th subcarrier. Meanwhile, the radar signal \(\mathbf{c}_l^{\mathrm{S}}(t) \in \mathbb{C}^{N N_{\text{RF}} \times 1}\), aggregated from all BSs, is defined as $
	\mathbf{c}_l^{\mathrm{S}}(t) = \big[ \mathbf{c}_{1,l}^{\mathrm{S}}(t)^\mathrm{T}, \ldots, \mathbf{c}_{N,l}^{\mathrm{S}}(t)^\mathrm{T} \big]^\mathrm{T} \sim \mathcal{CN}(\mathbf{0},\mathbf{S}_l)$, where $\mathbf{S}_l$ is the covariance matrix.

	\subsection{Near-field Secure Communication Model}
	{Following the  Rician channel model in~\cite{rician_ref,wideband_RIS_channel_estimation}, the channel response on the $l$-th subcarrier from the $n$-th BS to the $k$-th user is given by: }
	\begin{align}\label{channel_user}
		\mathbf{h}_{nk,l} = \sqrt{G_t} \mathbf{g}_l^n(\mathbf{l}_k^{\text{u}}) \odot \left(\psi_1 \mathbf{h}_{nk,l}^{\text{LoS}} + \psi_2 \mathbf{h}_{nk,l}^{\text{NLoS}}\right),
	\end{align}
	where $G_t > 0$ represents the antenna gain and the parameters $\psi_1 = \sqrt{\frac{k_{\mathrm{H}}}{1 + k_{\mathrm{H}}}}$ and $\psi_2 = \sqrt{\frac{1}{1 + k_{\mathrm{H}}}}$, with $ k_{\mathrm{H}} > 0 $ being the Rician factor. 
	The vector $\mathbf{g}_l^n(\mathbf{l}_k^{\text{u}}) = \big[ g_{1,l}^n(\mathbf{l}_k^{\text{u}}), \ldots, g_{N_\text{t},l}^n(\mathbf{l}_k^{\text{u}}) \big]^{\mathsf{T}}$ models the path loss between the $n$-th BS and the $k$-th user on the $l$-th subcarrier, where 
	$
	g_{n_\text{t},l}^n(\mathbf{l}_k^{\text{u}}) = \frac{\lambda_l}{4\pi \|\mathbf{l}_{n_{\mathrm{t}}}^{\text{b},n} - \mathbf{l}_k^{\text{u}} \|} e^{-\frac{1}{2}k_\text{abs}(f_l)\|\mathbf{l}_{n_{\mathrm{t}}}^{\text{b},n} - \mathbf{l}_k^{\text{u}} \|},
	$
	with $\lambda_l = c/f_l$ denoting the wavelength at subcarrier frequency $f_l$, and $c$ being the speed of light. Besides, \( k_\text{abs}(f_l) \) is the frequency-dependent molecule absorption coefficient \cite{b10}.
	{
	The LoS component \(\mathbf{h}_{nk,l}^{\text{LoS}}\) is modeled as 
	$ \mathbf{h}_{nk,l}^{\text{LoS}} = \left[ e^{\frac{-2 j \pi f_l}{c} \|\mathbf{l}_{1}^{\text{b},n} - \mathbf{l}_k^{\text{u}} \|}, \dots, e^{\frac{-2 j \pi f_l}{c} \|\mathbf{l}_{N_{\mathrm{t}}}^{\text{b},n} - \mathbf{l}_k^{\text{u}} \|} \right]^{\mathrm{T}}, $ 
	where \(f_l\) denotes the frequency of the \(l\)-th subcarrier, rendering the LoS component frequency-dependent. 
	The NLoS component \(\mathbf{h}_{nk,l}^{\text{NLoS}} \sim \mathcal{CN}(\mathbf{0}, \mathbf{I})\) is modeled as statistically independent across subcarriers for modeling simplicity~\cite{rician_ref,wideband_RIS_channel_estimation}.
	}

	The signal \( y_{k,l}^{\mathrm{u}}(t) \) received by the \(k\)-th user on the \(l\)-th subcarrier is given by
	\begin{align}
		y_{k,l}^{\mathrm{u}}(t) = \sum_{n=1}^{N} \mathbf{h}_{nk,l}^\mathrm{H} \bm{\chi}_{n,l}(t) + z_{k,l}(t),
	\end{align}
	where \( z_{k,l}(t) \sim \mathcal{C}\mathcal{N}(0, \sigma_{k}^2) \) represents the \gls{AWGN} at the \(k\)-th user.
	The achievable rate of the \( k \)-th user on the \(l\)-th subcarrier is given by
	\begin{align}
		r_{k,l}^\mathrm{u} \!\!=\!\! \log_2 \left( 1 \!+\! \frac{ \left| \mathbf{h}_{k,l}^{\mathrm{H}} \mathbf{F}_l \mathbf{d}_{k,l} \right|^2 }{ \sum_{j \neq k}^{K} \left| \mathbf{h}_{k,l}^{\mathrm{H}} \mathbf{F}_l \mathbf{d}_{j,l} \right|^2 \!\!+\!\!  \mathbf{h}_{k,l}^{\mathrm{H}} \mathbf{F}_l \mathbf{S}_l \mathbf{F}_l^{\mathrm{H}} \mathbf{h}_{k,l} \! \!+\!\! \sigma_{k}^2 }\right),
	\end{align}
	where 
	\( \mathbf{h}_{k,l} = [ \mathbf{h}_{1k,l}^{\mathrm{T}}, \ldots, \mathbf{h}_{Nk,l}^{\mathrm{T}} ]^{\mathrm{T}} \in \mathbb{C}^{N N_\mathrm{t} \times 1} \), \( \mathbf{d}_{k,l} = [ \mathbf{d}_{1k,l}^\mathrm{T}, \ldots, \mathbf{d}_{Nk,l}^\mathrm{T} ]^\mathrm{T} \in \mathbb{C}^{N N_\mathrm{RF} \times 1} \), and \( \mathbf{F}_l = \mathrm{blkdiag}(\mathbf{F}_{1,l}, \ldots, \mathbf{F}_{N,l}) \in \mathbb{C}^{N N_\mathrm{t} \times N N_\mathrm{RF}} \). 
	
	On the other hand, the signal received at the \( m \)-th target on the \( l \)-th subcarrier is given by
	\begin{equation}
		y_{m,l}^{\mathrm{t}}(t) = \sum_{n=1}^{N} \mathbf{a}_{nm,l}^\mathrm{H} \bm{\chi}_{n,l}(t) + z_{m,l}(t),
	\end{equation}
	where \( z_{m,l}(t) \sim \mathcal{CN}(0, \sigma_{m}^2) \) represents the \gls{AWGN} at the \( m \)-th target and 
	\( \mathbf{a}_{nm,l} \in \mathbb{C}^{N_\mathrm{t} \times 1} \) is the channel response between the \( n \)-th BS and the \( m \)-th target on the \( l \)-th subcarrier. 
	The channel response is modeled as:
	\begin{equation}
		\mathbf{a}_{nm,l} = \sqrt{G_t} \, \boldsymbol{\alpha}_l^n(\mathbf{l}_m^{\mathrm{t}}) \odot \left( \psi_1 \mathbf{a}_{nm,l}^{\text{LoS}} + \psi_2 \mathbf{a}_{nm,l}^{\text{NLoS}} \right),
	\end{equation}
	where \( \boldsymbol{\alpha}_l^n(\mathbf{l}_m^{\mathrm{t}}) = \left[ \alpha_{1,l}^n(\mathbf{l}_m^{\mathrm{t}}), \dots, \alpha_{N_\mathrm{t},l}^n(\mathbf{l}_m^{\mathrm{t}}) \right]^\mathrm{T} \) is the path loss vector between the \( n \)-th BS and the \( m \)-th target. 
	The terms \( \mathbf{a}_{nm,l}^{\text{LoS}} \) and \( \mathbf{a}_{nm,l}^{\text{NLoS}} \) correspond to the \gls{LOS} and \gls{NLoS} components, respectively, and are modeled analogously to (\ref{channel_user}).

	
	The $m$-th target may act as a potential eavesdropper attempting to intercept the confidential information intended for the $k$-th user. Accordingly, its achievable eavesdropping rate on the $l$-th subcarrier is given by
	\begin{equation}
		r_{mk,l}^\mathrm{t} = \log_2 \left( 1 + \mathrm{SINR}_{mk,l}^\mathrm{t} \right),
	\end{equation}
	where the corresponding SINR is expressed as
	\begin{equation}
		\mathrm{SINR}_{mk,l}^\mathrm{t} =
		\frac{
			\left| \mathbf{a}_{m,l}^\mathrm{H} \mathbf{F}_l \mathbf{d}_{k,l} \right|^2
		}{
			\sum\limits_{j \neq k}^{K} \left| \mathbf{a}_{m,l}^\mathrm{H} \mathbf{F}_l \mathbf{d}_{j,l} \right|^2
			+ \mathbf{a}_{m,l}^\mathrm{H} \mathbf{F}_l \mathbf{S}_l \mathbf{F}_l^\mathrm{H} \mathbf{a}_{m,l}
			+ \sigma_{m}^2
		},
	\end{equation}
	with $\mathbf{a}_{m,l}=\left[ \mathbf{a}_{1m,l}^{\mathrm{T}}, \ldots, \mathbf{a}_{Nm,l}^{\mathrm{T}} \right]^\mathrm{T} \in \mathbb{C}^{NN_\mathrm{t} \times 1}$ denoting the aggregated steer vector from all \glspl{BS} to the $m$-th target on the $l$-th subcarrier.
	
	\subsection{Near-field Sensing Model}
	The steering vector $\mathbf{b}_{m,l} \in \mathbb{C}^{N_\mathrm{r} \times 1}$ between the receiving \gls{BS} and the $m$-th target on the \( l \)-th subcarrier is modeled as:
	\begin{align}
		\mathbf{b}_{m,l} = \sqrt{G_r} \, \boldsymbol{\beta}_l(\mathbf{l}_m^{\mathrm{t}}) \odot \left( \psi_1 \mathbf{b}_{m,l}^{\text{LoS}} + \psi_2 \mathbf{b}_{m,l}^{\text{NLoS}} \right),
	\end{align}
	where $G_r > 0$ denotes the receiver antenna gain. The terms $\boldsymbol{\beta}_l(\mathbf{l}_m^{\mathrm{t}})$ (path loss), $\mathbf{b}_{m,l}^{\text{LoS}}$ (LoS component), and $\mathbf{b}_{m,l}^{\text{NLoS}}$ (NLoS component) follow the same modeling approach as in (\ref{channel_user}).

	
	For notational simplicity, we define the following matrices: \( \mathbf{B}_l = [ \mathbf{b}_{1,l},  \ldots, \mathbf{b}_{M,l}] \in \mathbb{C}^{N_\mathrm{r} \times M}\), \( \mathbf{A}_l = [ \mathbf{a}_{1,l}, \ldots, \mathbf{a}_{M,l} ]^\mathrm{T} \in \mathbb{C}^{M \times N N_t}\), 
	and $\mathbf{D}_l = [ \mathbf{d}_{1,l}, \ldots, \mathbf{d}_{K,l}] \in \mathbb{C}^{N N_\mathrm{RF} \times K}$.
	The received signal at the receiving \gls{BS} is given by
	\begin{align}
		\mathbf{y}_{l}^\mathrm{r}(t) = \mathbf{B}_l \mathbf{Q} \mathbf{A}_l \mathbf{F}_l \left[\mathbf{D}_l,\mathbf{I}_{NN_{\mathrm{RF}}}\right]  \mathbf{c}_l(t)  + \mathbf{z}_{l}^\mathrm{r}(t),
	\end{align}
	where $\mathbf{Q} = \mathrm{diag}(q_1, q_2, \ldots, q_M) \in \mathbb{C}^{M \times M}$ denotes the \gls{RCS} matrix for all $M$ targets, 
	$\mathbf{c}_l^\mathrm{I}(t) = [c_{1,l}^\mathrm{I}(t), \ldots, c_{K,l}^\mathrm{I}(t)]^\mathrm{T}\sim\mathcal{CN}\left(\mathbf{0},\mathbf{I}_K\right)$, $\mathbf{c}_l(t)\triangleq\left[ \mathbf{c}_l^\mathrm{I}(t);\mathbf{c}_l^\mathrm{S}(t)\right]$,
	and $\mathbf{z}_{l}^{\mathrm{r}}(t) \sim \mathcal{CN}(\mathbf{0}, \sigma_\mathrm{r}^2 \mathbf{I})$ is the AWGN at the receiving \gls{BS}.
	
	To locate multiple targets, we apply matched filtering across $P$ received OFDM symbols, obtaining an $N_\mathrm{r} \times (K + N N_\mathrm{RF})$-dimensional data matrix:
	\begin{align}\label{Y_r^l}
		\mathbf{Y}_l^\mathrm{r} = \mathbf{B}_l \mathbf{Q} \mathbf{A}_l \mathbf{F}_l [ \mathbf{D}_l, \mathbf{I}_{NN_{\mathrm{RF}}} ]	\frac{1}{P}\sum_{p=1}^P\mathbf{c}_l(p)\mathbf{c}_l(p)^\mathrm{H} + \mathbf{N}_\mathrm{s}^l,
	\end{align}
	where 
	$
	\mathbf{N}_\mathrm{s}^l = \frac{1}{P} \sum_{p=1}^{P} \mathbf{z}_{l}^{\mathrm{r}}[p]
	\mathbf{c}_l^\mathrm{H}[p].$ Besides, the probability distribution of $\mathbf{Y}_l^\mathrm{r}$ is given in the following lemma. 
	\begin{lemma}\label{noise}
		When $P$ is sufficiently large, 
		$\mathrm{vec}\left(\mathbf{Y}_l^\mathrm{r}\right)\sim \mathcal{C}\mathcal{N} \left( \boldsymbol{\mu}_l, \boldsymbol{\Xi}_l \right)$, where
		$\boldsymbol{\Xi}_l = \frac{\sigma_\mathrm{r}^2}{P} \, \text{blkdiag}\left( \mathbf{I}_K, \mathbf{S}_l^* \right) \otimes \mathbf{I}_{N_\mathrm{r}}$,
	\end{lemma}
	\begin{IEEEproof}
		The proof is given in Appendix A. 
	\end{IEEEproof}
	According to Lemma~\ref{noise}, $\mathrm{vec}\!\left(\mathbf{Y}_l^\mathrm{r}\right)$ follows a complex Gaussian distribution \cite{ELSAA_ISAC} with mean vector
	$
	\boldsymbol{\mu}_l
	=
	\left[
	\bm{\varpi}_{1,l};
	\ldots;
	\bm{\varpi}_{K,l};
	\bm{\varsigma}_{1,l};
	\ldots;
	\bm{\varsigma}_{NN_{\mathrm{RF}},l}
	\right],
	$
	where
	$
	\bm{\varpi}_{k,l}\triangleq\mathbf{B}_l \mathbf{Q} \mathbf{A}_l \mathbf{F}_l \mathbf{d}_{k,l},
	$
	and
	$
	\bm{\varsigma}_{k,l}\triangleq\mathbf{B}_l \mathbf{Q} \mathbf{A}_l \mathbf{F}_l \mathbf{S}_{l}(:,k).
	$
	By stacking the received signals across all subcarriers \cite{b6}, we further obtain
	\begin{align}
		\left[\mathrm{vec}\left(\mathbf{Y}_1^\mathrm{r}\right);\ldots;\mathrm{vec}\left(\mathbf{Y}_L^\mathrm{r}\right)\right]
		\sim
		\mathcal{CN}\left(\boldsymbol{\mu}, \boldsymbol{\Xi}\right),
	\end{align}
	where
	$
	\boldsymbol{\mu}=
	[\boldsymbol{\mu}_1^\mathrm{T},\dots,\boldsymbol{\mu}_L^\mathrm{T}]^\mathrm{T}
	$
	and
	$
	\boldsymbol{\Xi}=
	\mathrm{blkdiag}\left(\boldsymbol{\Xi}_1,\dots,\boldsymbol{\Xi}_L\right)
	$.
	We adopt the \gls{CRB} to evaluate the locating performance. The target-related parameter vector is defined as
	$
	\boldsymbol{\theta} = \left[\mathbf{x}^\mathrm{t},\, \mathbf{y}^\mathrm{t},\, \mathbf{z}^\mathrm{t},\, \mathbf{q}_\mathrm{R},\, \mathbf{q}_\mathrm{I} \right]^\mathrm{T} \in \mathbb{R}^{5M \times 1},
	$
	where $\mathbf{x}^\mathrm{t} = [x_1^\mathrm{t}, \ldots, x_M^\mathrm{t}]$, $\mathbf{y}^\mathrm{t} = [y_1^\mathrm{t}, \ldots, y_M^\mathrm{t}]$, and $\mathbf{z}^\mathrm{t} = [z_1^\mathrm{t}, \ldots, z_M^\mathrm{t}]$ represent the 3D Cartesian coordinates of the $M$ targets. In addition, $\mathbf{q}_\mathrm{R} = [\Re(q_1), \ldots, \Re(q_M)]$ and $\mathbf{q}_\mathrm{I} = [\Im(q_1), \ldots, \Im(q_M)]$ denote the real and imaginary parts of the \gls{RCS} of the targets, respectively.
	
	The analytical result of the \gls{FIM} associated with the parameter vector $\boldsymbol{\theta}$ is given by the following theorem.

	\begin{theorem}\label{cauchy}
		The FIM $\mathcal{I}(\boldsymbol{\theta},\boldsymbol{\theta}) \in \mathbb{R}^{5M \times 5M}$ is given by
		\renewcommand{\arraystretch}{1.3}
		\begin{align}
			&{\mathcal{I}(\boldsymbol{\theta},\boldsymbol{\theta})}=  \notag\\
			&2\begin{bmatrix}
				\Re(\mathcal{I}_{\mathbf{x}^\text{t}\mathbf{x}^\text{t}}) & \Re(\mathcal{I}_{\mathbf{x}^\text{t}\mathbf{y}^\text{t}}) & \Re(\mathcal{I}_{\mathbf{x}^\text{t}\mathbf{z}^\text{t}}) & \Re(\mathcal{I}_{\mathbf{x}^\text{t}\mathbf{q}}) & -\Im(\mathcal{I}_{\mathbf{x}^\text{t}\mathbf{q}})  \\
				\Re(\mathcal{I}_{\mathbf{x}^\text{t}\mathbf{y}^\text{t}}^\mathrm{T}) & \Re(\mathcal{I}_{\mathbf{y}^\text{t}\mathbf{y}^\text{t}}) & \Re(\mathcal{I}_{\mathbf{y}^\text{t}\mathbf{z}^\text{t}}) & \Re(\mathcal{I}_{\mathbf{y}^\text{t}\mathbf{q}}) & -\Im(\mathcal{I}_{\mathbf{y}^\text{t}\mathbf{q}})  \\
				\Re(\mathcal{I}_{\mathbf{x}^\text{t}\mathbf{z}^\text{t}}^\mathrm{T}) & \Re(\mathcal{I}_{\mathbf{y}^\text{t}\mathbf{z}^\text{t}}^\mathrm{T}) & \Re(\mathcal{I}_{\mathbf{z}^\text{t}\mathbf{z}^\text{t}}) & \Re(\mathcal{I}_{\mathbf{z}^\text{t}\mathbf{q}}) & -\Im(\mathcal{I}_{\mathbf{z}^\text{t}\mathbf{q}}) \\
				\Re(\mathcal{I}_{\mathbf{x}^\text{t}\mathbf{q}}^\mathrm{T}) & \Re(\mathcal{I}_{\mathbf{y}^\text{t}\mathbf{q}}^\mathrm{T}) & \Re(\mathcal{I}_{\mathbf{z}^\text{t}\mathbf{q}}^\mathrm{T}) & \Re(\mathcal{I}_{\mathbf{q}\mathbf{q}}) & -\Im(\mathcal{I}_{\mathbf{q}\mathbf{q}}) \\
				-\Im(\mathcal{I}_{\mathbf{x}^\text{t}\mathbf{q}}^\mathrm{T}) & -\Im(\mathcal{I}_{\mathbf{y}^\text{t}\mathbf{q}}^\mathrm{T}) & -\Im(\mathcal{I}_{\mathbf{z}^\text{t}\mathbf{q}}^\mathrm{T}) & -\Im(\mathcal{I}_{\mathbf{q}\mathbf{q}}^\mathrm{T}) & \Re(\mathcal{I}_{\mathbf{q}\mathbf{q}})
			\end{bmatrix},\nonumber
		\end{align}
	\end{theorem}
	where $\mathbf{q}$ stands for $\mathbf{q}_\text{R}$ to simplify notation.
	The analytical expressions of the terms $\mathcal{I}_{\mathbf{u}\mathbf{v}}$, for $\mathbf{u}, \mathbf{v} \in \{\mathbf{x}, \mathbf{y}, \mathbf{z}\}$, are provided in \eqref{f_uv} at the top of the next page. In addition, the expressions of $\mathcal{I}_{\mathbf{u}\mathbf{q}}$ and $\mathcal{I}_{\mathbf{q}\mathbf{q}}$ are given by \eqref{fuq} and \eqref{fqq} in Appendix~B, respectively.
	Here, the derivative matrices involved in \eqref{f_uv} are specified as 
$\dot{\mathbf{B}_l}_\mathbf{u} = \left[ \frac{\partial \mathbf{b}_{1,l}}{\partial u_1}, \frac{\partial \mathbf{b}_{2,l}}{\partial u_2}, \ldots, \frac{\partial \mathbf{b}_{M,l}}{\partial u_M} \right]$ and
		$\dot{\mathbf{A}_l}_\mathbf{u} = \left[ \frac{\partial \mathbf{a}_{1,l}}{\partial u_1}, \frac{\partial \mathbf{a}_{2,l}}{\partial u_2}, \ldots, \frac{\partial \mathbf{a}_{M,l}}{\partial u_M} \right]^\mathrm{T}$.
	The definitions of $\dot{\mathbf{B}_l}_{\mathbf{v}}$ and $\dot{\mathbf{A}_l}_{\mathbf{v}}$ are analogous to those of $\dot{\mathbf{B}_l}_{\mathbf{u}}$ and $\dot{\mathbf{A}_l}_{\mathbf{u}}$.
		The vector $\mathbf{s}_{\kappa,l} \in \mathbb{C}^{NN_\mathrm{RF} \times 1}$ corresponds to the $\kappa$-th column of the signal covariance matrix $\mathbf{S}_l$. Additionally, $\tilde{s}_{\kappa g,l} \in \mathbb{C}$ denotes the $(\kappa, g)$-th entry of $(\mathbf{S}_l^*)^{-1}$.
	\begin{IEEEproof}
		The proof is given in Appendix B. 
	\end{IEEEproof}
	Accordingly, the CRB matrix of $\boldsymbol{\theta}$ is $\mathbf{C} = \mathcal{I}^{-1}(\boldsymbol{\theta},\boldsymbol{\theta})$. Then, the CRB for locating the $m$-th target is $\text{CRB}_m = \sum_{i=0}^{2} \mathbf{C}[m + iM,\, m + iM]$, where $\mathbf{C}[m, m]$, $\mathbf{C}[m+M, m+M]$ and $\mathbf{C}[m+2M, m+2M]$ correspond to the CRB for estimating its $x$, $y$, and $z$-coordinates, respectively  \cite{ELSAA_ISAC}.

	\begin{figure*}[!ht]
		\vspace{-8mm}
		\setlength{\jot}{0pt}
		\begin{align}\label{f_uv}
			&\mathcal{I}_{\mathbf{u}\mathbf{v}}= \frac{P}{\sigma_\mathrm{r}^2} \Big\{\sum_{l=1}^{L}
			\big({\dot{\mathbf{B}_l}_\mathbf{u}}^\mathrm{H} \dot{\mathbf{B}_l}_\mathbf{v} \big)  \odot \big(\mathbf{Q} \mathbf{A}_l \mathbf{F}_l\mathbf{D}_l \mathbf{D}_l^\mathrm{H} \mathbf{F}_l^\mathrm{H} \mathbf{A}_l^\mathrm{H}\mathbf{Q}^\mathrm{H}  \big)^{\mathrm{T}}
			+ 
			\big( {\dot{\mathbf{B}_l}_\mathbf{u}}^\mathrm{H} \mathbf{B}_l \big)    \odot \big( \mathbf{Q} \dot{\mathbf{A}_l}_\mathbf{v} \mathbf{F}_l\mathbf{D}_l \mathbf{D}_l^\mathrm{H} \mathbf{F}_l^\mathrm{H} \mathbf{A}_l^\mathrm{H} \mathbf{Q}^\mathrm{H}\big)^{\mathrm{T}}   \notag \\
			&\ \ \ \ \ \ \ + 
			\big( \mathbf{B}_l^\mathrm{H} \dot{\mathbf{B}_l}_\mathbf{v}  \big) \odot \big(   \mathbf{Q} \mathbf{A}_l \mathbf{F}_l\mathbf{D}_l \mathbf{D}_l^\mathrm{H} \mathbf{F}_l^\mathrm{H} {\dot{\mathbf{A}_l}_\mathbf{u}}^\mathrm{H} \mathbf{Q}^\mathrm{H}  \big)^{\mathrm{T}} 
			+ 
			\big( \mathbf{B}_l^\mathrm{H} \mathbf{B}_l  \big) \odot  \big( \mathbf{Q} \dot{\mathbf{A}_l}_\mathbf{v} \mathbf{F}_l\mathbf{D}_l \mathbf{D}_l^\mathrm{H}\mathbf{F}_l^\mathrm{H}{\dot{\mathbf{A}_l}_\mathbf{u}}^\mathrm{H} \mathbf{Q}^\mathrm{H}\big)^{\mathrm{T}}   \notag \\ 
			& \ \ \ +   \sum_{l=1}^{L} \sum_{g=1}^{NN_\mathrm{RF}} \sum_{\kappa=1}^{NN_\mathrm{RF}}\tilde{s}_{\kappa g,l}  \Big( \big( \mathbf{Q} \mathbf{A}_l \mathbf{F}_l\mathbf{s}_{g,l} 	\mathbf{s}_{\kappa,l}^\mathrm{H}\mathbf{F}_l^\mathrm{H} \mathbf{A}_l^\mathrm{H} \mathbf{Q}^\mathrm{H}\big)^\mathrm{T} \odot \big( {\dot{\mathbf{B}_l}_\mathbf{u}}^\mathrm{H} \dot{\mathbf{B}_l}_\mathbf{v} \big)
			+ 
			\big( \mathbf{Q}\dot{\mathbf{A}_l}_\mathbf{v} \mathbf{F}_l\mathbf{s}_{g,l} 	\mathbf{s}_{\kappa,l}^\mathrm{H} \mathbf{F}_l^\mathrm{H} \mathbf{A}_l^\mathrm{H} \mathbf{Q}^\mathrm{H} \big)^\mathrm{T}  \odot \big( {\dot{\mathbf{B}_l}_\mathbf{u}}^\mathrm{H} 	\mathbf{B}_l \big) \notag \\
			&\ \ \ \ \ \ \ + 
			\big(  \mathbf{Q} \mathbf{A}_l \mathbf{F}_l\mathbf{s}_{g,l} 	\mathbf{s}_{\kappa,l}^\mathrm{H} \mathbf{F}_l^\mathrm{H} {\dot{\mathbf{A}_l}_\mathbf{u}}^\mathrm{H}\mathbf{Q}^\mathrm{H} \big)^\mathrm{T}  \odot \big(  \mathbf{B}_l^\mathrm{H}\dot{\mathbf{B}_l}_\mathbf{v}   \big)
			+ 
			\big(  \mathbf{Q} \dot{\mathbf{A}_l}_\mathbf{v} \mathbf{F}_l \mathbf{s}_{g,l} 	\mathbf{s}_{\kappa,l}^\mathrm{H} \mathbf{F}_l^\mathrm{H}  {\dot{\mathbf{A}_l}_\mathbf{u}}^\mathrm{H}\mathbf{Q}^\mathrm{H} \big)^\mathrm{T}   \odot \big(  \mathbf{B}_l^\mathrm{H} 	\mathbf{B}_l\big)\Big)  \Big\}.  
		\end{align}
		\hrule
	\end{figure*}

	In the following, we jointly design $\widetilde{\mathbf{F}}_n$, $\mathbf{T}_{n,l}$, $\mathbf{D}_l$, and $\mathbf{S}_l$ to maximize the minimum secrecy rate across all subcarriers and users under the maximal \gls{CRB} constraint, i.e.,
	{
	\begin{align} 
		&\max_{\widetilde{\mathbf{F}}_n, \mathbf{T}_{n,l}, \mathbf{D}_l, \mathbf{S}_l} \min_{k,l} \left( r_{k,l}^\mathrm{u} - \max_{m} r_{mk,l}^\mathrm{t} \right)^+ \notag\\
		&\text{s.t.}  \quad \text{CRB}_{m} \leqslant \Omega, \quad \forall m \in \{1,\ldots, M \}, \notag \\ 
		&\sum_{l=1}^{L} \left(\operatorname{tr}\left( \mathbf{F}_{n,l} \mathbf{D}_{n,l}\mathbf{D}_{n,l}^\mathrm{H} \mathbf{F}_{n,l}^\mathrm{H}  \right) \!+\!  \operatorname{tr}\left(\mathbf{F}_{n,l}\mathbf{E}_n\mathbf{S}_{l}\mathbf{E}_n^\text{H}\mathbf{F}_{n,l}^\mathrm{H} \right)\right) \!\leqslant\! P_{n},  \notag \\ 
		&\left|[\widetilde{\mathbf{F}}_n]_{i,j}\right| = 1,\quad \forall [\widetilde{\mathbf{F}}_n]_{i,j}\neq 0,\ \forall n \in \{1,\ldots,N\}, \notag\\
		&0\leq \tau_{g,n_{\mathrm{RF}}} \leq \tau_\text{max},  \forall g \in \{1,\ldots,G\},  \forall n_{\mathrm{RF}} \in \{1,\ldots,N_{\mathrm{RF}}\},\label{constraint problem}
	\end{align}
	}
	where $\mathbf{E}_n\triangleq\left[\overset{n-1}{\overbrace{\mathbf{0}_{N_{\mathrm{RF}}},\ldots\mathbf{0}_{N_{\mathrm{RF}}}}},\mathbf{I}_{N_{\mathrm{RF}}},\overset{N-n}{\overbrace{\mathbf{0}_{N_{\mathrm{RF}}},\ldots\mathbf{0}_{N_{\mathrm{RF}}}}}\right]$, $\mathbf{F}_{n,l}\triangleq\widetilde{\mathbf{F}}_n\mathbf{T}_{n,l}$, $\Omega$ is the maximum tolerable CRB, $P_n$ denotes the transmit power budget of the $n$-th BS, and $\tau_{\max}$ denotes the maximum time delay supported by each TTD module. {The objective in (\ref{constraint problem}) is to improve the minimum secrecy performance among all user--subcarrier pairs, so as to prevent certain links from suffering excessively weak secrecy performance and becoming bottlenecks of the overall system.}

	It is challenging to solve problem (\ref{constraint problem}) for three main reasons. First, the nonconvex objective function and CRB constraint preclude the use of efficient convex optimization algorithm. Second, multiplicative relationships between variables introduce nonconvexity and multiple local optima, complicating gradient-based optimization. Third, the high dimensionality of the optimization variables results in prohibitive computational complexity.
	
	\section{GNN-Enabled Joint Optimization Approach}
	{ As illustrated in Fig.~\ref{section_Fig}, to accommodate the structured interactions among multiple entities and the varying system configurations in the considered cooperative \gls{THz}-\gls{ISAC} system}, we propose a Mamba-empowered \gls{GNN} to perform the joint design, which consists of four functional blocks: (\uppercase\expandafter{\romannumeral1}) {Message passing mechanism}: A neighborhood aggregation mechanism facilitates the propagation of channel-related information among nodes. At each iteration, new messages are generated by aggregating both incoming and outgoing edge information, with iterative updates across the graph; (\uppercase\expandafter{\romannumeral2}) {Message generation module}: An MLP-based block denoted as $\mathcal{F}(\cdot)$ is employed to extract features from the aggregated messages and generate new messages for subsequent processing; (\uppercase\expandafter{\romannumeral3}) {Mamba module}: {Leveraging the selection mechanism of state space models (SSMs), 	the Mamba blocks $\mathcal{G}(\cdot)$ further enhance the modeling of longer-range task-related dependencies beyond local neighborhood interactions in a scalable manner, facilitating cooperative ISAC system design}.

	
	\begin{figure*} [ht]
		\vspace{-4mm}
		\centering
		\includegraphics[width=1\linewidth]{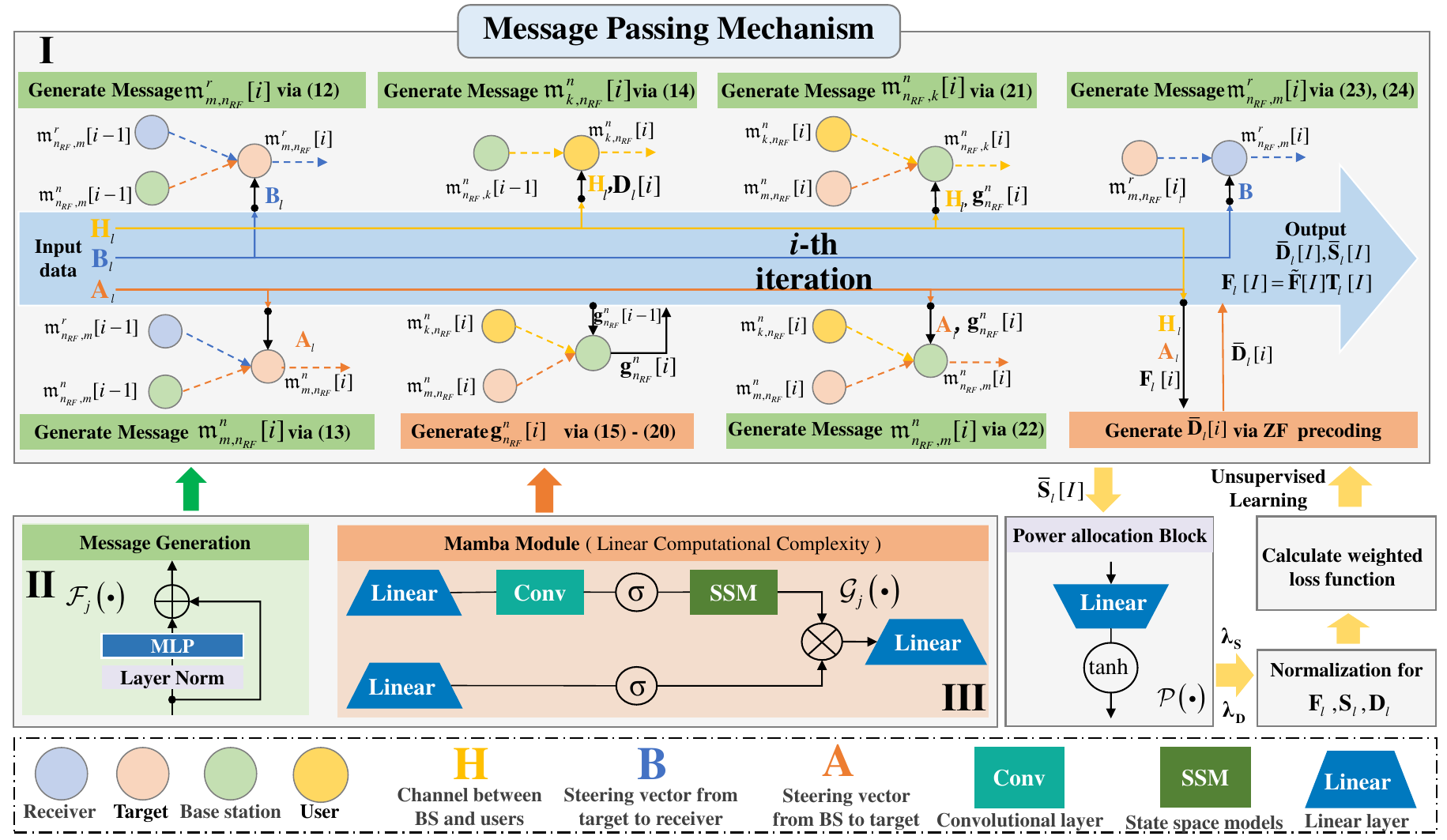}

		\caption{Proposed Mamba-empowered GNN for hybrid beamforming and radar signal design. In this figure, we take the $l$-th subcarrier as an example. The aggregated TTD matrix is defined as \( \mathbf{T}_l = \mathrm{blkdiag}(\mathbf{T}_{1,l}, \ldots, \mathbf{T}_{N,l}) \in \mathbb{C}^{ N G N_\mathrm{RF} \times N N_\mathrm{RF}} \). $\overline{\mathbf{S}}_l$ and ${\mathbf{S}}_l$ denote the radar signal covariance matrices before and after power normalization. $\overline{\mathbf{D}}_l$ and ${\mathbf{D}}_l$ denote the digital precoding matrices before and after power normalization. Since the phase shifter matrix \( \widetilde{\mathbf{F}} \) is identical across all subcarriers, the multiplication \( \mathbf{F}_l = \widetilde{\mathbf{F}} \mathbf{T}_l \) is used for notational simplicity.}
		\vspace{-1mm}\label{section_Fig}
	\end{figure*}

	\subsection{Graph Representation}\label{sec Graph Representation}

	\subsubsection{Vertices and their feature vectors}
	The graph consists of four types of vertices: user vertices $\{v^u_k \mid \forall k \in \mathcal{K}\}$, target vertices $\{v^t_m \mid \forall m \in \mathcal{M}\}$, BS RF-chain vertices $\{v^n_{n_{\text{RF}}} \mid \forall n \in \mathcal{N}, n_{\text{RF}} \in \mathcal{N}_{\text{RF}}\}$, and receiver RF-chain vertices $\{v^r_{n_{\text{RF}}} \mid \forall n_{\text{RF}} \in \mathcal{N}_{\text{RF}}\}$, where $\mathcal{K}\triangleq\left\{1,\ldots,K\right\}$, $\mathcal{M}\triangleq\left\{1,\ldots,M\right\}$,  $\mathcal{N}\triangleq\left\{1,\ldots,{N}\right\}$, and $\mathcal{N}_{\text{RF}}\triangleq\left\{1,\ldots,{N}_{\text{RF}}\right\}$.

	Since the design of the PS matrix $\widetilde{\mathbf{F}}\triangleq\mathrm{blkdiag}\left(\widetilde{\mathbf{F}}_{1},\ldots,\widetilde{\mathbf{F}}_{N}\right)$, TDD matrix $\mathbf{T}_l\triangleq\mathrm{blkdiag}\left(\widetilde{\mathbf{T}}_{1,l},\ldots,\widetilde{\mathbf{T}}_{N,l}\right)$, and radar covariance matrix $\mathbf{S}_l$ occurs at the BS, we construct feature vectors only for the BS RF-chain vertices. Specifically, each BS RF-chain node $v^n_{n_{\text{RF}}}$ has a feature vector $\mathbf{g}^n_{n_{\text{RF}}} \in \mathbb{R}^{(2J_t + G + 2N N_{\text{RF}}) \times 1}$ capturing the sufficient information for designing \gls{ISAC} precoder at the \gls{BS}, where $J_t = N_t / N_{\text{RF}}$ is the number of antennas per RF chain. This vector comprises: (i) The aggregated phase shifter vector $\boldsymbol{\phi}_{n_{\text{RF}}}^n \in \mathbb{C}^{J_t \times 1}$ is formed by stacking the phase vectors from the $G$ TTD branches associated with a single RF chain; (ii) the delay vector $\boldsymbol{\tau}_{n_{\text{RF}}}^n \in \mathbb{R}^{G \times 1}$, used to compute the frequency-dependent TTD vector $\mathbf{t}_{l,n_{\text{RF}}}^n$; and (iii) the radar waveform covariance vector $\mathbf{s}^n_{n_{\text{RF}}} \in \mathbb{C}^{N N_{\text{RF}} \times 1}$, used to construct $\mathbf{S}_l$. The complete feature vector is given by $\mathbf{g}^n_{n_{\text{RF}}} = [\Re(\boldsymbol{\phi}_{n_{\text{RF}}}^n)^\mathrm{T}, \Im(\boldsymbol{\phi}_{n_{\text{RF}}}^n)^\mathrm{T}, (\boldsymbol{\tau}_{n_{\text{RF}}}^n)^\mathrm{T}, \Re(\mathbf{s}^n_{n_{\text{RF}}})^\mathrm{T}, \Im(\mathbf{s}^n_{n_{\text{RF}}})^\mathrm{T}]^\mathrm{T}$. Combining these features across all BS RF-chains enables the construction of $\widetilde{\mathbf{F}}$, $\mathbf{T}_l$, and $\mathbf{S}_l$.
	
	\subsubsection{Edges and their feature vectors}
	
	{
	The graph is constructed according to the interaction relationship of the considered system, where only physically meaningful links are retained.
	}
	Specifically, each user vertex $v^u_k$ connects to all BS RF-chain vertices $v^n_{n_{\text{RF}}}$ via an edge $e^{n}_{k,n_{\text{RF}}}$, representing the channel between the $k$-th user and the $n_{\text{RF}}$-th RF chain of the $n$-th BS. Similarly, edges $e^{n}_{m,n_{\text{RF}}}$ and $e^{r}_{m,n_{\text{RF}}}$ model the channels from the $m$-th target to the $n_{\text{RF}}$-th RF chain of the $n$-th transmitting BS and  receiving \gls{BS}, respectively. 
	Their real-valued feature vectors are:
	$
		\mathbf{x}^{n,e}_{k,n_{\text{RF}}} = \left[\Re\left(\mathbf{h}^{n}_{k,n_{\text{RF}}}\right)^{\mathrm{T}}, \Im\left(\mathbf{h}^{n}_{k,n_{\text{RF}}}\right)^{\mathrm{T}}\right]^{\mathrm{T}} \in \mathbb{R}^{2J_t \times 1}, $
	$	\mathbf{x}^{n,e}_{m,n_{\text{RF}}} = \left[\Re\left(\mathbf{a}^{n}_{m,n_{\text{RF}}}\right)^{\mathrm{T}}, \Im\left(\mathbf{a}^{n}_{m,n_{\text{RF}}}\right)^{\mathrm{T}}\right]^{\mathrm{T}} \in \mathbb{R}^{2J_t \times 1},$
	and $	\mathbf{x}^{r,e}_{m,n_{\text{RF}}} = \left[\Re\left( \mathbf{b}^{r}_{m,n_{\text{RF}}}\right)^{\mathrm{T}}, \Im\left( \mathbf{b}^{r}_{m,n_{\text{RF}}}\right)^{\mathrm{T}}\right]^{\mathrm{T}} \in \mathbb{R}^{2J_r \times 1}$. 
	{As a result, the computational cost of  subsequent message passing is governed by the number of retained edges, instead of growing quadratically with the number of nodes as in a fully connected graph.}

	\subsection{Proposed GNN Structure} \label{GNN Structure}
	
In the graph, message exchanges occur between:
\begin{itemize}
	\item User-Transmitting \gls{BS} RF chain pairs: $\mathfrak{m}^n_{k,n_{\text{RF}}}(v^u_k \rightarrow v^n_{n_{\text{RF}}})$ and $\mathfrak{m}^n_{n_{\text{RF}},k}(v^n_{n_{\text{RF}}} \rightarrow v^u_k)$
	\item Target-Transmitting \gls{BS} RF chain pairs: $\mathfrak{m}^n_{m,n_{\text{RF}}}(v^t_m \rightarrow v^n_{n_{\text{RF}}})$ and $\mathfrak{m}^n_{n_{\text{RF}},m}(v^n_{n_{\text{RF}}} \rightarrow v^t_m)$
	\item Target-Receiving \gls{BS}  RF chain pairs: $\mathfrak{m}^r_{m,n_{\text{RF}}}(v^t_m \rightarrow v^r_{n_{\text{RF}}})$ and $\mathfrak{m}^r_{n_{\text{RF}},m}(v^r_{n_{\text{RF}}} \rightarrow v^t_m)$
\end{itemize}

In each GNN iteration, message passing constructs the RF chain vertex features $\mathbf{g}^n_{n_{\text{RF}}}$ to optimize:  
(1) the PS matrix $\widetilde{\mathbf{F}}$,  
(2) the TTD matrix $\mathbf{T}_l$, and  
(3) the radar covariance matrix (prior to power normalization) $\overline{\mathbf{S}}_l$.  
Subsequently, to eliminate multiuser interference, we adopt the ZF precoding scheme from~\cite{ZFHP} to design the baseband precoder $\overline{\mathbf{D}}_l$ (prior to power normalization), based on $\widetilde{\mathbf{F}}$ and $\mathbf{T}_l$. This approach improves training efficiency by allowing the focus to remain on optimizing $\widetilde{\mathbf{F}}$, $\mathbf{S}_l$, and $\mathbf{T}_l$.
Finally, we normalize ${\mathbf{D}}_l$ and ${\mathbf{A}}_l$, and then use a Power Allocation Block (PAB) to distribute power  between communication and sensing.
As illustrated in Fig.~\ref{section_Fig}, the GNN performs message passing via neighborhood aggregation. 
The messages $\{\mathfrak{m}^{r}_{m, n_{\text{RF}}}[i], \mathfrak{m}^{n}_{m, n_{\text{RF}}}[i], \mathfrak{m}^{n}_{k, n_{\text{RF}}}[i], \mathfrak{m}^{r}_{n_{\text{RF}}, k}[i], \mathfrak{m}^{n}_{n_{\text{RF}}, m}[i], \mathfrak{m}^{r}_{n_{ \text{RF}}, m}[i]\}$ are generated sequentially with MLP block $\mathcal{F}_1(\cdot) \sim \mathcal{F}_6(\cdot)$. Using neighborhood messages $\mathfrak{m}^{n}_{k, n_{\text{RF}}}, \mathfrak{m}^{n}_{m, n_{\text{RF}}}$ and prior states $\widetilde{\mathbf{F}}[i-1], \mathbf{T}_{l}[i-1], \overline{\mathbf{S}}_{l}[i-1]$, the Mamba-based blocks $\mathcal{G}_1(\cdot) \sim \mathcal{G}_3(\cdot)$ generate $\widetilde{\mathbf{F}}[i], \mathbf{T}_{l}[i], \overline{\mathbf{S}}_{l}[i]$. 

Details are given as follows:
	\subsubsection {Generating messages $\mathfrak{m}^{r}_{m,n_{\text{RF}}}[i]$ and $\mathfrak{m}^{n}_{m,n_{\text{RF}}}[i]$}\label{F_1 F_2}

The target vertex $v^t_m$ aggregates messages from its neighboring vertices representing both transmitting and receiving \gls{BS} chains:
Transmitting BS messages $\left(\{\mathfrak{m}^{n}_{n_{\text{RF}},m}[i-1], \forall n \in \mathcal{N}, \forall n_{\text{RF}} \in \mathcal{N}_{\text{RF}}\}\right)$ and receiving \gls{BS} messages $\left(\{\mathfrak{m}^{r}_{n_{\text{RF}},m}[i-1], \forall n_{\text{RF}} \in \mathcal{N}_{\text{RF}}\}\right)$.
The aggregation is defined by
	$\mathfrak{m}_m[i-1] \triangleq \mathcal{S}\left( \{\mathfrak{m}^{n}_{n_{\text{RF}},m}[i-1]\} \right) \in \mathbb{R}^{D\times 1}$ and
	$\mathfrak{m}^r_m[i-1] \triangleq \mathcal{S}\left( \{\mathfrak{m}^{r}_{n_{\text{RF}},m}[i-1]\} \right) \in \mathbb{R}^{D\times 1}$,
where $\mathcal{S}(\cdot)$ is an order-invariant sum operator (insensitive to $N$ and $N_{RF}$ ordering) and $D$ is the hidden dimension.

These aggregated messages are then processed using $\mathcal{F}_1(\cdot)$ and $\mathcal{F}_2(\cdot)$ with corresponding edge features:
	\begin{align}
			\label{F_1}
			\mathfrak{m}^{r}_{m,n_{\text{RF}}}[i]
			=&\mathcal{F}_1\left(\mathfrak{m}_m[i-1],\mathfrak{m}^r_m[i-1],\mathbf{x}^{r,e}_{m,n_{\text{RF}}} \right),\\
			\label{F_2}
			\mathfrak{m}^{n}_{m,n_{\text{RF}}}[i]
			=&\mathcal{F}_2\left(\mathfrak{m}_m[i-1],\mathfrak{m}^r_m[i-1],\mathbf{x}^{n,e}_{m,n_{\text{RF}}} \right).
		\end{align}

\subsubsection {Generating message $\mathfrak{m}^n_{k,n_{\text{RF}}}[i]$}

The user vertex $v^u_k$ aggregates messages 
$\{\mathfrak{m}^{n}_{n_{\text{RF}},k}[i-1], \forall n \in \mathcal{N}, \forall n_{\text{RF}} \in \mathcal{N}_{\text{RF}}\}$ 
	from all RF chain vertices of $N$ BSs, i.e.,
$\mathfrak{m}_k[i-1] = \mathcal{S}\left( \{\mathfrak{m}^{n}_{n_{\text{RF}},k}[i-1], \forall n \in \mathcal{N}, \forall n_{\text{RF}} \in \mathcal{N}_{\text{RF}}\} \right)$,
where $\mathfrak{m}_k[i-1] \in \mathbb{R}^{D \times 1}$. 
	The message $\mathfrak{m}^n_{k,n_{\text{RF}}}[i]$ is generated as:
\begin{align}
	\label{F_3}
	\mathfrak{m}^{n}_{k,n_{\text{RF}}}[i] &= \mathcal{F}_3 \Big( \mathfrak{m}_k[i-1], \, \mathbf{x}^{n,e}_{k,n_{\text{RF}}}, \, \mathbf{d}^n_{n_{\text{RF}}}[i-1] \Big),
\end{align}
where $\mathbf{d}^n_{n_{\text{RF}}}[i-1] = \big[ \Re(d^n_{n_{\text{RF}}}[i-1]), \, \Im(d^n_{n_{\text{RF}}}[i-1]) \big]^\mathrm{T} \in \mathbb{R}^{2}$ 
denotes the digital beamforming vector from the previous iteration, and 
$\big\{ d^{n}_{n_{\text{RF}}} \in \mathbb{C} \mid \forall n \in \mathcal{N}, \, \forall n_{\text{RF}} \in \mathcal{N}_{\text{RF}} \big\}$ 
represents the $(n-1) N_{\text{RF}} + n_{\text{RF}}$-th element of $\mathbf{d}_{k,l}$.

\subsubsection{Generating RF Chain Features $\mathbf{g}_{n,n_{\text{RF}}}[i]$} \label{Gen_fs}
The vertex $v^n_{n_{\text{RF}}}$ aggregates messages from neighboring target and user vertices:
\begin{align}
	\mathfrak{m}^n_{n_{\text{RF}}}[i] &\triangleq \mathcal{S}\left( \{\mathfrak{m}^{n}_{m,n_{\text{RF}}}[i], \forall m \in \mathcal{M} \} \right), \label{M_mtoBS} \\  
	\hat{\mathfrak{m}}^n_{n_{\text{RF}}}[i] &\triangleq \mathcal{S}\left( \{\mathfrak{m}^n_{k,n_{\text{RF}}}[i], \forall k \in \mathcal{K} \} \right), \label{M_ktoBS}
\end{align}
where $\mathfrak{m}^n_{n_{\text{RF}}}[i], \hat{\mathfrak{m}}^n_{n_{\text{RF}}}[i] \in \mathbb{R}^{D \times 1}$ are dimension-invariant to $M$ and $K$. 

To decouple optimization variables $\widetilde{\mathbf{F}}$, $\mathbf{T}_{l}$, and $\overline{\mathbf{S}}_{l}$, vertex $v^n_{n_{\text{RF}}}$ uses three Mamba blocks $\mathcal{G}_1(\cdot)$, $\mathcal{G}_2(\cdot)$, and $\mathcal{G}_3(\cdot)$ to generate $\mathbf{g}^n_{n_{\text{RF}}}[i]$, incorporating the previous state $\mathbf{g}^n_{n_{\text{RF}}}[i-1]$ for iterative refinement. The updates are:
\begin{align}
	\hat{\mathbf{g}}^n_{n_{\text{RF}}}[i] &= \mathcal{G}_1\left(\mathfrak{m}^n_{n_{\text{RF}}}[i], \hat{\mathfrak{m}}^n_{n_{\text{RF}}}[i], \hat{\mathbf{g}}^n_{n_{\text{RF}}}[i-1] \right), \label{g_hat} \\
	\bar{\mathbf{g}}^n_{n_{\text{RF}}}[i] &= \mathcal{G}_2\left(\mathfrak{m}^n_{n_{\text{RF}}}[i], \hat{\mathfrak{m}}^n_{n_{\text{RF}}}[i], \bar{\mathbf{g}}^n_{n_{\text{RF}}}[i-1] \right), \label{g_bar} \\
	\tilde{\mathbf{g}}^n_{n_{\text{RF}}}[i] &= \mathcal{G}_3\left(\mathfrak{m}^n_{n_{\text{RF}}}[i], \hat{\mathfrak{m}}^n_{n_{\text{RF}}}[i], \hat{\mathbf{g}}^n_{n_{\text{RF}}}[i], \bar{\mathbf{g}}^n_{n_{\text{RF}}}[i], \tilde{\mathbf{g}}^n_{n_{\text{RF}}}[i-1] \right), \label{g_tilde}
\end{align}
where $\hat{\mathbf{g}}^n_{n_{\text{RF}}}[i] \triangleq [\Re(\boldsymbol{\phi}_{gn_{\mathrm{RF}}})^{\mathrm{T}}[i], \Im(\boldsymbol{\phi}_{gn_{\mathrm{RF}}})^{\mathrm{T}}[i]]^{\mathrm{T}} \in \mathbb{R}^{2J_t \times 1}$ (used for constructing $\tilde{\mathbf{F}}$), $\bar{\mathbf{g}}^n_{n_{\text{RF}}}[i] \triangleq \boldsymbol{\tau}_{n_{\mathrm{RF}}}^n \in \mathbb{R}^{G \times 1}$ (used for constructing $\mathbf{T}_l$), and $\tilde{\mathbf{g}}^n_{n_{\text{RF}}}[i] \triangleq [\Re(\mathbf{s}^n_{n_{\text{RF}}})^{\mathrm{T}}[i], \Im(\mathbf{s}^n_{n_{\text{RF}}})^{\mathrm{T}}[i]]^\text{T} \in \mathbb{R}^{2NN_{\text{RF}} \times 1}$ (used for constructing $\overline{\mathbf{S}}_l$).

The final feature vector is concatenated as:
\begin{align}
	\mathbf{g}_{n_{\text{RF}}}^n[i] = \left[ (\hat{\mathbf{g}}^n_{n_{\text{RF}}})^\mathrm{T}[i], (\bar{\mathbf{g}}^n_{n_{\text{RF}}})^\mathrm{T}[i], (\tilde{\mathbf{g}}^n_{n_{\text{RF}}})^\mathrm{T}[i] \right]^\mathrm{T}.\label{Final_Feature}
\end{align}
	\subsubsection {Generating messages $\mathfrak{m}^n_{n_{\text{RF}},k}[i]$ and $\mathfrak{m}^{n}_{n_{\text{RF}},m}[i]$}
The messages $\mathfrak{m}^n_{n_{\text{RF}},k}[i]$ and $\mathfrak{m}^n_{n_{\text{RF}},m}[i]$ are generated by the operators $\mathcal{F}_4(\cdot)$ and $\mathcal{F}_5(\cdot)$, respectively, as follows:
	\begin{align}
		\label{F_4}
		\mathfrak{m}^n_{n_{\text{RF}},k}[i] &= \mathcal{F}_4\left(\mathfrak{m}^n_{n_{\text{RF}}}[i], \hat{\mathfrak{m}}^n_{n_{\text{RF}}}[i], \mathbf{g}^n_{n_{\text{RF}}}[i], \mathbf{x}^{n,e}_{k,n_{\text{RF}}}\right), \\
		\label{F_5}
		\mathfrak{m}^n_{n_{\text{RF}},m}[i] &= \mathcal{F}_5\left(\mathfrak{m}^n_{n_{\text{RF}}}[i], \hat{\mathfrak{m}}^n_{n_{\text{RF}}}[i], \mathbf{g}^n_{n_{\text{RF}}}[i], \mathbf{x}^{n,e}_{m,n_{\text{RF}}}\right),
	\end{align}

{\subsubsection {Generating message $\mathfrak{m}^{r}_{n_{\text{RF}},m}[i]$}
	The vertex $v^r_{n_{\text{RF}}}$ first receives messages  $\{\mathfrak{m}^{r}_{m,n_{\text{RF}}}[i],\forall m \in \mathcal{M} \}$ from the $M$ target vertices $\{v^t_m, \forall m \in \mathcal{M} \}$. The aggregated message is
	\begin{align}
		\mathfrak{m}^r_{n_{\text{RF}}}[i] \triangleq \mathcal{S}\left( {\mathfrak{m}^{r}_{m,n_{\text{RF}}}[i] \mid \forall m \in \mathcal{M} } \right). \label{M_mtoReceiver}
	\end{align}
	Subsequently, the message for each target vertex is generated using $\mathfrak{m}^r_{n_{\text{RF}}}[i]$ and the edge feature ${\mathbf{x}^{r,e}_{m,n_{\text{RF}}}}$:
	\begin{align}
		\label{F_6}
		\mathfrak{m}^{r}_{n_{\text{RF}},m}[i] = \mathcal{F}_6\left(\mathfrak{m}^r_{n_{\text{RF}}}[i], \mathbf{x}^{r,e}_{m,n_{\text{RF}}} \right).
	\end{align}

	}
	\begin{algorithm}[t]	
		\setstretch{0.8}
		\caption{Proposed Message Passing Algorithm of GNN.}
		\label{algorithm1}
		\begin{algorithmic}[1]
			\INPUT Channel matrices $\mathbf{H}_l$, $\mathbf{A}_l$, and $\mathbf{B}_l$; 
			
			Initializing analog precoder $\mathbf{F}_l$, digital precoder $\overline{\mathbf{D}}_l$, and radar signal covariance matrix $\overline{\mathbf{S}}_l$. 
			\FOR {$i = 1,2,\cdots,I$}        
			\STATE Generate messages form the targets to the BSs and receiver via Eqs. (\ref{F_1}) and (\ref{F_2}).   
			\STATE Generate messages form the users to the BSs via Eq. (\ref{F_3}).   
			\STATE Update $\tilde{\mathbf{F}}[i]$, $\mathbf{T}_l[i]$, and $\overline{\mathbf{S}}_l[i]$ via (\ref{g_hat}), (\ref{g_bar}) and (\ref{g_tilde}).   
			\STATE Generate messages form the BSs to users and targets via Eqs. (\ref{F_4}) and (\ref{F_5}).
			\STATE Generate messages form the receiving BS to the targets via Eq. (\ref{F_6}).
			\STATE Update $\overline{\mathbf{D}}_l[i]$ by ZF precoding.
			\ENDFOR
			\OUTPUT Normalize $\overline{\mathbf{D}}_l[i]$, and $\overline{\mathbf{S}}_l[i]$, and then perform the power allocation using $\mathcal{P}(\cdot)$..
			
		\end{algorithmic}
	\end{algorithm}
	\subsubsection{Designing digital precoder $\overline{\mathbf{D}}_l$}
To enhance communication reliability and security, we design the ZF precoding matrix  (prior to power normalization) $\overline{\mathbf{D}}_l$ to simultaneously eliminate inter-user interference and suppress signal leakage to potential eavesdroppers. 
Given the analog beamforming matrix $\mathbf{F}_l[i]$, we first construct the equivalent communication and sensing channels on the $l$-th subcarrier as:
$\mathbf{H}_{l}^{\text{eq}}[i] = \mathbf{H}_{l}^{*} \mathbf{F}_l[i]$ and $\mathbf{A}_{l}^{\text{eq}}[i] = \mathbf{A}_{l}^{*} \mathbf{F}_l[i]$.
To minimize information leakage, we perform singular value decomposition (SVD) on $\mathbf{A}_{l}^{\text{eq}}[i]$ to obtain its null-space matrix $\mathbf{V}_{l}[i] \in \mathbb{C}^{NN_{\text{RF}} \times (NN_{\text{RF}} - M)}$, which satisfies:
$
	\mathbf{A}_{l}^{\text{eq}}[i] \mathbf{V}_{l}[i] = \mathbf{0}.
$
The digital precoding matrix for the $l$-th subcarrier (prior to power normalization) can be designed as
$\overline{\mathbf{D}}_{l}[i] = \mathbf{V}_{l}[i] \left( \mathbf{H}_{l}^{\text{eq}}[i] \mathbf{V}_{l}[i] \right)^{\dagger}.$

After obtaining  $\mathbf{F}_l[I]$, $\overline{\mathbf{D}}_l[I]$, and $\overline{\mathbf{S}}_l[I]$, we perform the power allocation between communication signals and radar signals for satisfying the power budget of problem~\eqref{constraint problem}.  
	
		
	\subsubsection{Power Allocation Between Communication and Radar Signals}
	As shown in the right part of Fig.~\ref{section_Fig}, we first employ PAB, $\mathcal{P}(\cdot)$, to generate the weights $\lambda_{\mathbf{S}}$ and $\lambda_{\mathbf{D}}$ for ${\mathbf{S}}_l$ and ${\mathbf{D}}_{n,l}$, respectively. We first construct the radar covariance matrix. To ensure the Hermitian positive semidefinite property, the intermediate matrix is formed as $\overline{\mathbf{S}}_l \leftarrow \overline{\mathbf{S}}_l \overline{\mathbf{S}}_l^\mathrm{H}$. The normalized radar covariance matrix is then given by
	$
	{\mathbf{S}}_l = \frac{\lambda_{\mathbf{S}} \overline{\mathbf{S}}_l}{\sum_{j=1}^{L} \operatorname{Tr}(\overline{\mathbf{S}}_j)}.
	$
	Based on the resulting radar covariance matrices, the digital precoder for the $n$-th BS on the $l$-th subcarrier is normalized as
	$
	{\mathbf{D}}_{n,l} = \frac{\lambda_{\mathbf{D}} \overline{\mathbf{D}}_{n,l}}{J_t \sqrt{\sum_{j=1}^{L} \|\overline{\mathbf{D}}_{n,j}\|_\text{F}^2}},
	$
	where $J_t$ denotes the number of antennas per RF chain, and
	$
	\lambda_{\mathbf{D}}= \sqrt{P_n-\sum_{j=1}^{L} \operatorname{tr}\left(\mathbf{F}_{n,j}\mathbf{E}_n\mathbf{S}_{j}\mathbf{E}_n^\text{H}\mathbf{F}_{n,j}^\mathrm{H} \right)}.
	$

It is worth pointing out that the above power allocation results in the following problems:
		
			1) \textbf{Inconsistent Scaling:} Variations in the pre-normalization power levels $\|\overline{\mathbf{D}}_{n,l}\|_\text{F}^2$ across different BSs lead to non-uniform scaling factors.
			
			2) \textbf{Performance Impact:} Degrades interference suppression capability and reduces information leakage mitigation effectiveness.

	\textbf{Algorithm~\ref{algorithm1}} summarizes the \gls{GNN}-enabled joint optimization approach.

	\subsection{Proposed Graph Mamba Block}
	\begin{figure}
		\vspace{-4mm}
		\centering
		\includegraphics[width=1\linewidth]{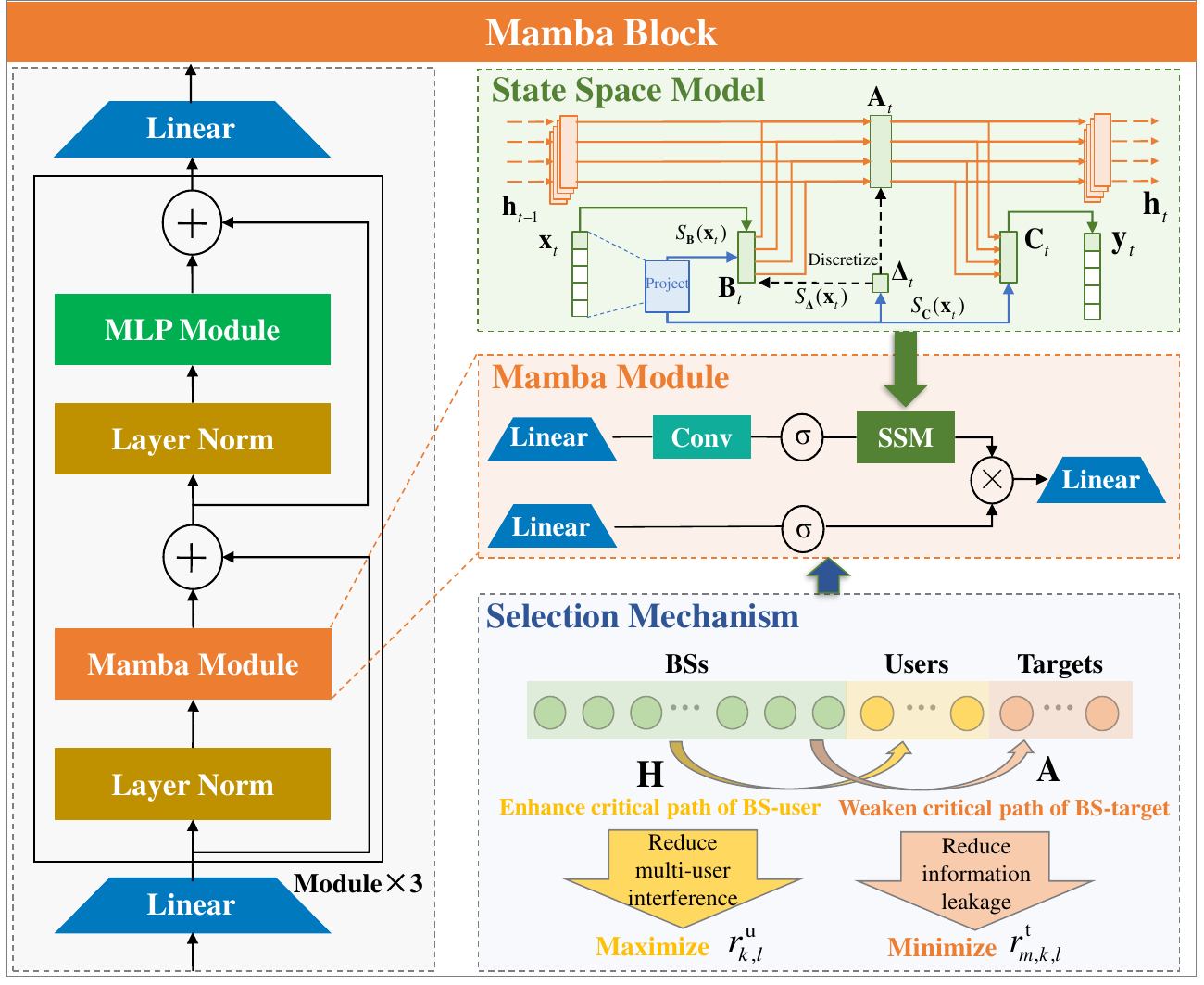}
		\caption{Detailed neural network structure of Mamba block.}
\label{Mamba Block}
	\end{figure}
	
In this subsection, we detail the architectures of the message generation blocks $\mathcal{F}(\cdot)$ and the Mamba blocks $\mathcal{G}(\cdot)$, which generate the PS ($\widetilde{\mathbf{F}}$), the TTD ($\mathbf{T}_l$), and the radar signal covariance matrix ($\overline{\mathbf{S}}_l$). The integration of the Mamba block \cite{mamba} enables adaptive feature extraction through its selection mechanism and learnable state-space model (SSM). 
	
	As illustrated in Fig.~\ref{Mamba Block}, each $\mathcal{G}(\cdot)$ block comprises three identical modules, each containing a Mamba module followed by an MLP module. The Mamba module adopts a dual-path neural network architecture. One branch consists of a sequentially connected linear layer followed by a Swish activation function \cite{swish}. The other branch is composed of a linear layer, a convolutional layer, a Swish activation, and a SSM \cite{mamba}. The outputs of the two branches are then combined via element-wise multiplication and passed through an additional linear layer to produce the final output. The MLP module, on the other hand, follows a standard sequential structure composed of a linear layer, a ReLU activation \cite{relu}, and another linear layer.
	On the other hand, the message generation block $\mathcal{F}(\cdot)$ adopts a simplified MLP-based structure to reduce computational complexity. It consists of a layer normalization followed by three sequential linear layers, with ReLU activations applied after the first two layers. Additionally, we adopt the $\tanh$ activation function within $\mathcal{P}(\cdot)$ because it helps mitigate vanishing gradients during back-propagation.
The input, hidden, and output dimensions of $\mathcal{F}_1(\cdot) \sim \mathcal{F}_6(\cdot)$, $\mathcal{G}_1(\cdot) \sim \mathcal{G}_3(\cdot)$, and $\mathcal{P}(\cdot)$ are detailed in Table~\ref{tab:block_dims}.

	{Compared with traditional model-based optimization and deep learning-based methods, the proposed Mamba-empowered GNN framework offers the following key advantages:
	\begin{itemize}
		\item  \emph{Feature Adaptivity:} The Mamba module contributes critical dynamic feature selection capability through its learnable SSM. Specifically, given the sparse structure of THz channel, the SSM adaptively tunes its parameters to select the most relevant features, concentrating the power on the main beam direction. This approach provides three key advantages:
	    1) Enhanced capture of dominant path components to improve communication performance for specific users;
	    2) Mitigation of multiuser interference and information leakage risks by dynamically filtering irrelevant features;
		3) Natural alignment with THz channels' low-rank nature through the Mamba architecture.		
		Furthermore, this selective mechanism enables fine-grained representation of sparse channels, making it particularly suitable for THz systems.
		\item \emph{Computational Efficiency:} By synergistically integrating the neighborhood aggregation mechanism of GNNs with Mamba's selective SSM, our framework enables efficient processing in high-dimensional  \gls{ELAA} \gls{ISAC} system design. Compared to traditional alternating optimization methods that suffer from high computational complexity, our approach offers three distinct advantages: 1) Maintains linear computational complexity, significantly reducing processing overhead; 2) Enhances real-time performance, achieving inference within $10^{-2}$ second; 
		\item \emph{Generalization and Scalability:} Compared with deep learning architectures that are typically associated with fixed input structures, the proposed Mamba-empowered \gls{GNN} framework is better suited to topology-varying environments. Owing to the permutation-invariant aggregation and structural inductive bias of \glspl{GNN}, the learned model does not rely on a topology-specific global mapping, but instead learns beamforming rules from compositions of local physical interactions, thereby enabling generalization across different network conditions. Furthermore, unlike conventional methods that require solving a new optimization problem for each network realization, the proposed approach generates the transmit design through direct forward inference, which allows it to scale more favorably to large-scale systems.
	\end{itemize}
	}

	\begin{table}[t]
		\centering
		\caption{Input-Output Dimensions of Mamba and MLP Blocks}
		\label{tab:block_dims}
		\begin{tabular}{ccc}
			\toprule
			\textbf{} & \textbf{Block} & \tabincell{c}{\textbf{Input, Hidden, Output} \\ \textbf{Dimensions}} \\
			\midrule
			\multirow{4}{*}{\shortstack{Mamba \\ blocks \\ $\times$ 3 }}
			& $ \mathcal{G}_1(\cdot)$ & $2(D+J_t)$, $2D$, $2J_t$ \\
			& $ \mathcal{G}_2(\cdot)$  & $2D + G$, $2D$, $G$ \\
			& $ \mathcal{G}_3(\cdot)$ & \tabincell{c}{$2D + 2J_t + G + 2NN_{\text{RF}}$, \\ $2D$, $2NN_{\text{RF}}$} \\ 
			\cmidrule(lr){1-3}
			
			\multirow{6}{*}{\shortstack{MLP \\ blocks \\ $\times$ 3}} 
			& $\mathcal{F}_1(\cdot)$ & $2(D+J_r)$, $2D$, $D$\\
			& $\mathcal{F}_2(\cdot)$ & $2(D+J_t)$, $2D$, $D$ \\
			& $\mathcal{F}_3(\cdot)$ & $D+2J_t+2$, $2D$, $D$ \\
			& $\mathcal{F}_4(\cdot)$ & $2(D+2J_t+NN_{\text{RF}})$, $2D$, $D$  \\
			& $\mathcal{F}_5(\cdot)$ & $2(D+2J_t+NN_{\text{RF}})$, $2D$, $D$  \\
			& $\mathcal{F}_6(\cdot)$ & $D+2$, $2D$, $D$  \\ 

			\cmidrule(lr){1-3}
			
			\multirow{1}{*}{\shortstack{PAB} } 
			& $ \mathcal{P}(\cdot)$ & $ 2NN_{\text{RF}}NN_{\text{RF}}$, $2D$, 1   \\
			
			\bottomrule
		\end{tabular}
	\end{table}
	\subsection{Neural Network Training}
	{To optimize the joint design of the hybrid beamforming matrix and radar signal covariance matrix, we adopt an unsupervised training framework for the proposed Mamba-empowered \gls{GNN}, which does not rely on precomputed optimal solutions as supervision. To jointly account for the secrecy-rate objective and the \gls{CRB} constraint\footnote{{Here, the \gls{CRB} is used only as an offline sensing-performance criterion during training. During deployment, the trained network determines the transmit configuration directly from the available \gls{CSI}, without requiring the exact target locations.}}, we construct a weighted loss function that balances secure communication and sensing accuracy, given by
	\begin{align}
		\mathcal{L}(\boldsymbol{\Theta}) 
		= \omega_1 \min_{k,l} \left( r_{k,l}^{\mathrm{u}} - \max_m r_{mk,l}^{\mathrm{t}} \right)
		+ \omega_2 \sum_m (\mathrm{CRB}_{m} - \Omega)^{+},\nonumber
	\end{align}
	where $\omega_1 < 0$ and $\omega_2 > 0$ are weighting parameters that control the trade-off between communication security and sensing accuracy. This trade-off arises from the intrinsic coupling between sensing and secure communication in the considered \gls{ISAC} system. In particular, improving sensing performance requires stronger signal interaction with the target, whereas secure communication seeks to limit the information exposed to the target. The network parameters are optimized by minimizing $\mathcal{L}(\boldsymbol{\Theta})$ via mini-batch stochastic gradient descent. 		During the end-to-end training process, the GNN module can gradually learn to aggregate sensing-related observations from distributed BSs through graph-based message passing, allowing target-location-dependent geometric cues embedded in the composite channel to be jointly propagated and fused. 
	Meanwhile, the Mamba module further models the sequential dependencies of these fused representations and selectively emphasizes the features most relevant to the CRB-based sensing constraints. 
	In this way, the proposed Mamba-GNN enhance geometry-related sensing features that are mainly governed by LoS propagation, while reducing the influence of irregular fluctuations induced by NLoS scattering on 3D target sensing.}
	
	\subsection{Complexity Analysis}
		
	In each iteration, the main computational components include message aggregation and generation, the Mamba block, and the zero-forcing algorithm. Since message aggregation involves only additive operations, its computational cost is negligible compared to multiplicative operations. The key operations are as follows:
	\begin{enumerate}[1)]
		\item Message generation is primarily dominated by the operations of a MLP, with a computational complexity of $\mathcal{O}(NL(K+M+N_{\text{RF}})D^{2})$.
		\item The Mamba block incurs a computational complexity of $\mathcal{O}(NLN_{\text{RF}}D(J_{\text{t}} + G + NN_{\text{RF}}))$.
		\item The computational complexity of the ZF algorithm is dominated by the pseudo-inversion of the equivalent channel matrix, which has a complexity of $\mathcal{O}(MN^2)$.
	\end{enumerate}
	{Given the condition $K, M \ll NN_{\text{RF}} \ll NN_{\text{t}}$, the overall computational complexity of the proposed method is approximately $\mathcal{O}(NLN_{\text{RF}}(D^{2}+(J_{\text{t}}+NN_{\text{RF}})D + KNN_{\text{t}}))$, which is significantly lower than that of the conventional alternating optimization approach \cite{Alt-Min}, whose complexity scales as $\mathcal{O}(LKN^{2}N_{\text{t}}^{2} + L^{3}K^{3}N_{\text{t}}N^{2}N_{\text{RF}}^{2})$. In particular, the complexity grows linearly with the user number $K$ and does not grow quadratically with the total number of graph nodes, which allows the proposed framework to scale more favorably in larger cooperative systems.}
	
	\begin{figure*}[t]
	\vspace{-4mm}
	\centering
	\subfloat[\footnotesize{GNN vs. Alt-Min with varying $L$, where $N_{\text{RF}}{=}16$.}]
	{\includegraphics[width=0.5\linewidth]{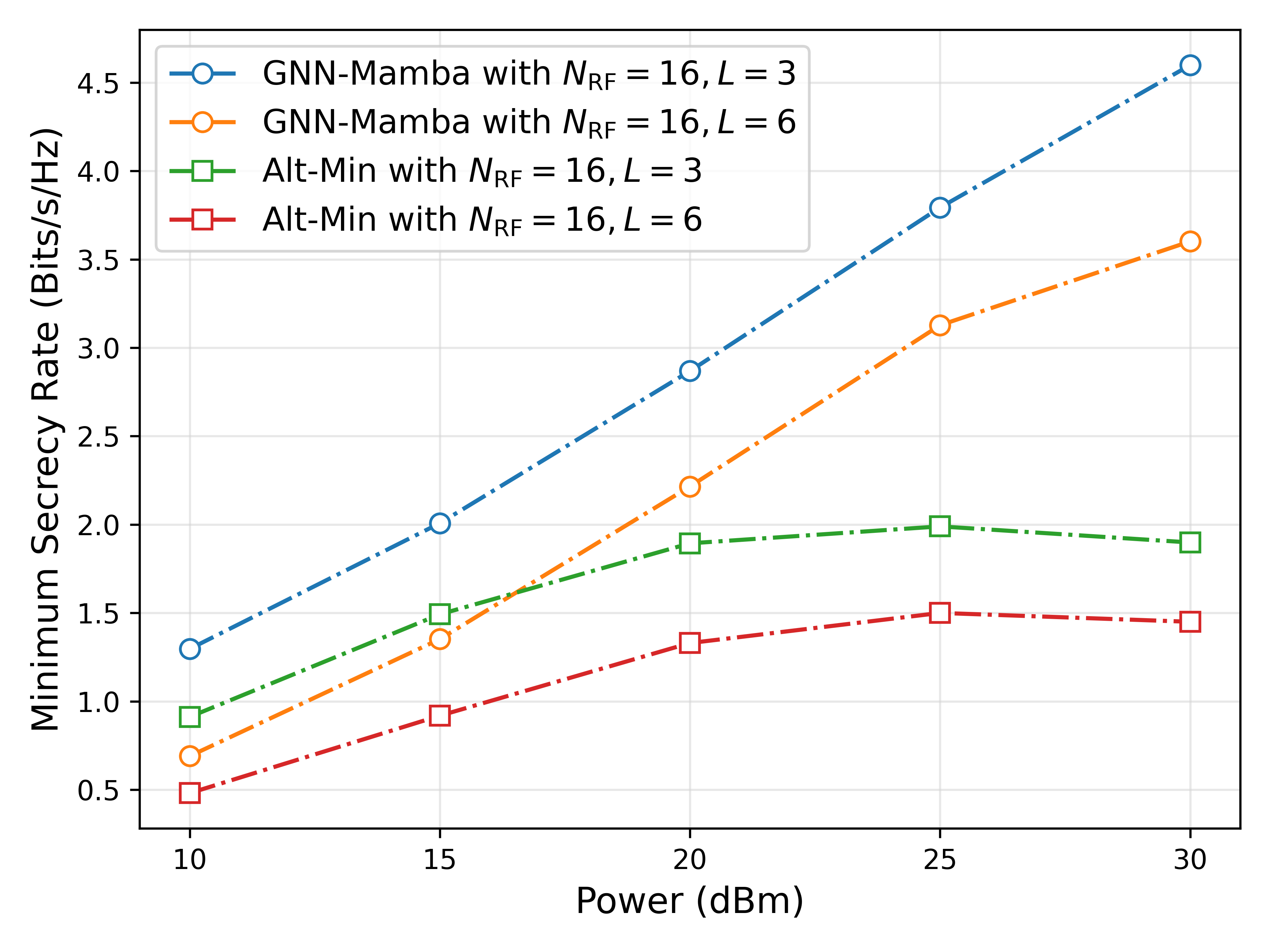}\label{Simul_1}}
	\hfill
	\subfloat[\footnotesize{GNN vs. Alt-Min with varying $N_{\text{RF}}$, where $L{=}3$.}]
	{\includegraphics[width=0.5\linewidth]{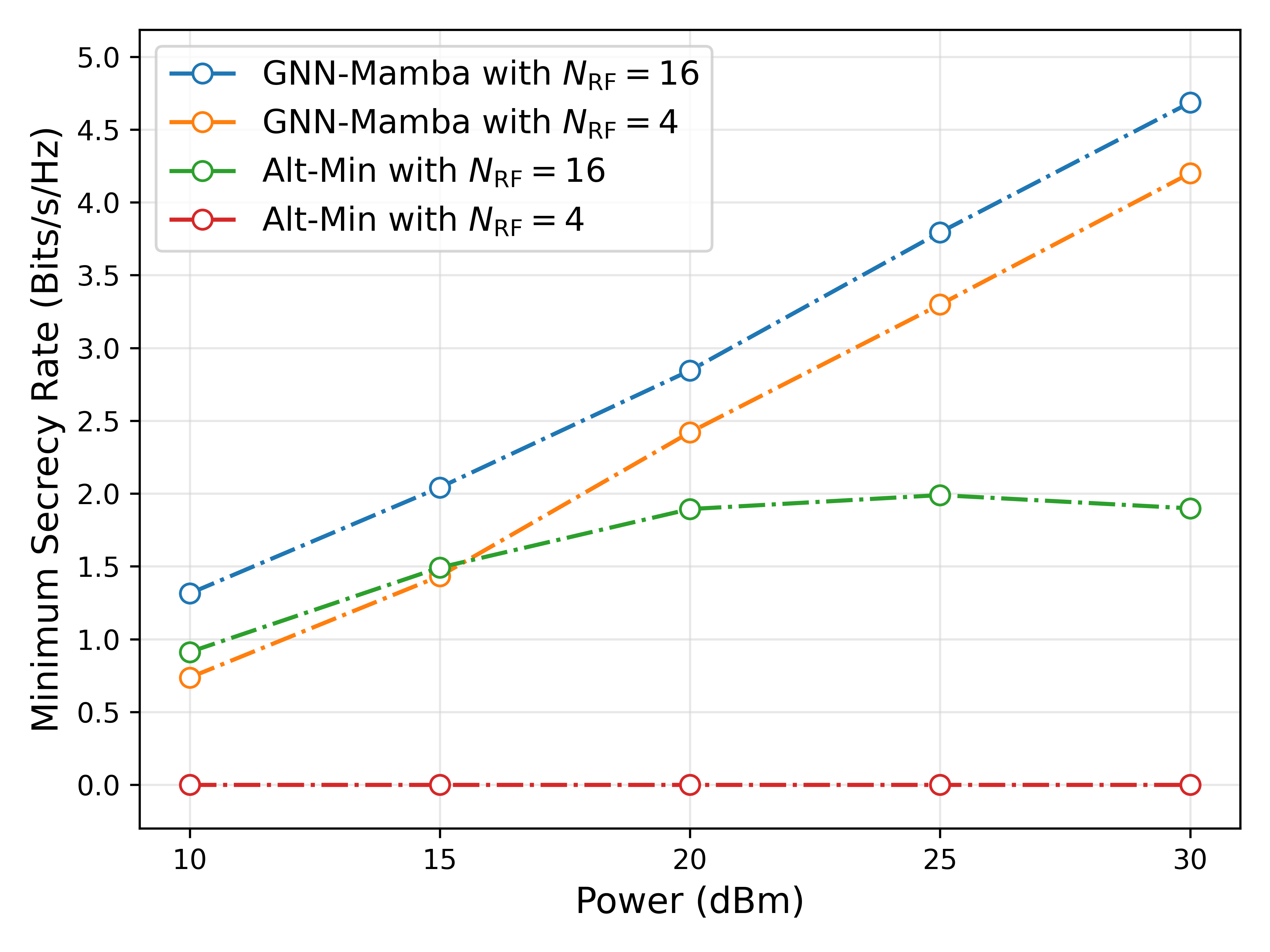}\label{Simul_2}}
	
	\vspace{-3mm}
	
	\subfloat[\footnotesize{GNN performance with $M{=}3$ and varying $K$, where $N_{\text{RF}}{=}4$ and $L{=}3$.}]
	{\includegraphics[width=0.5\linewidth]{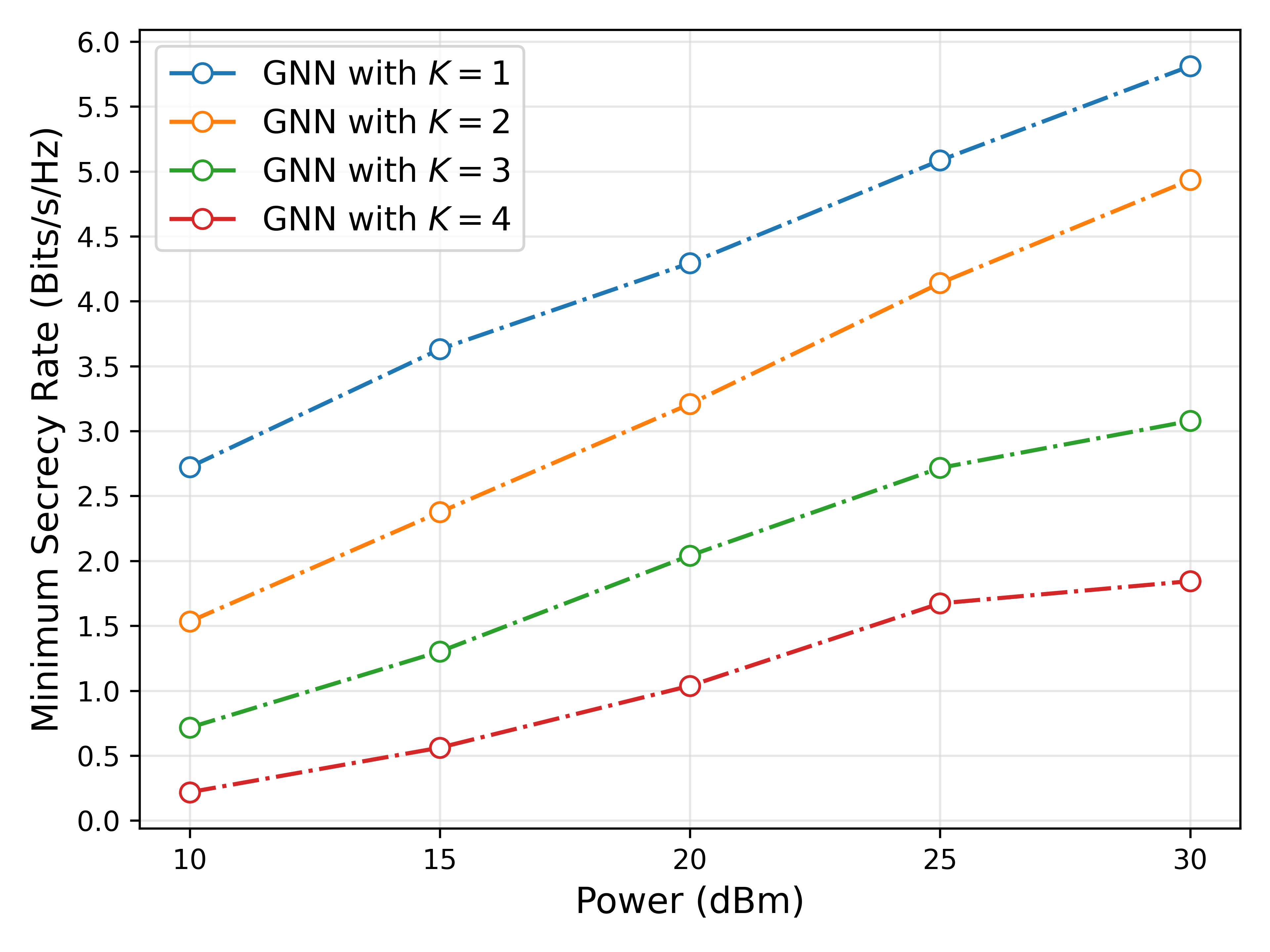}\label{Simul_3}}
	\hfill
	\subfloat[\footnotesize{GNN performance with $K{=}3$ and varying $M$, where $N_{\text{RF}}{=}4$ and $L{=}3$.}]
	{\includegraphics[width=0.5\linewidth]{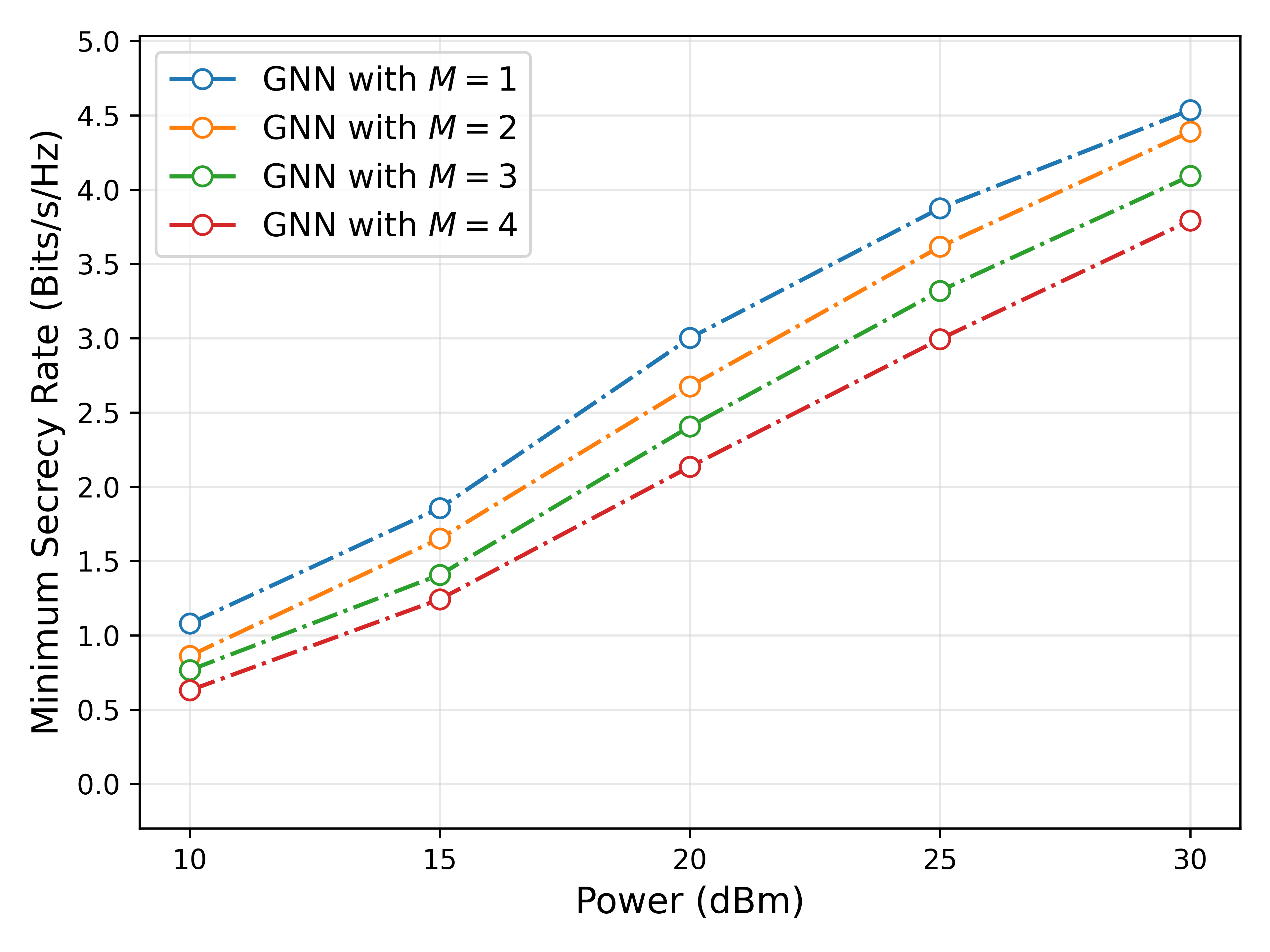}\label{Simul_4}}
	
	\vspace{-1mm}
	\caption{Performance comparison of GNN and Alt-Min over varying transmit power.}
	\label{fig:scalability_gnn}
	\vspace{-3mm}
\end{figure*}

	\begin{figure}[t]
	\centering
	\includegraphics[width=1\linewidth]{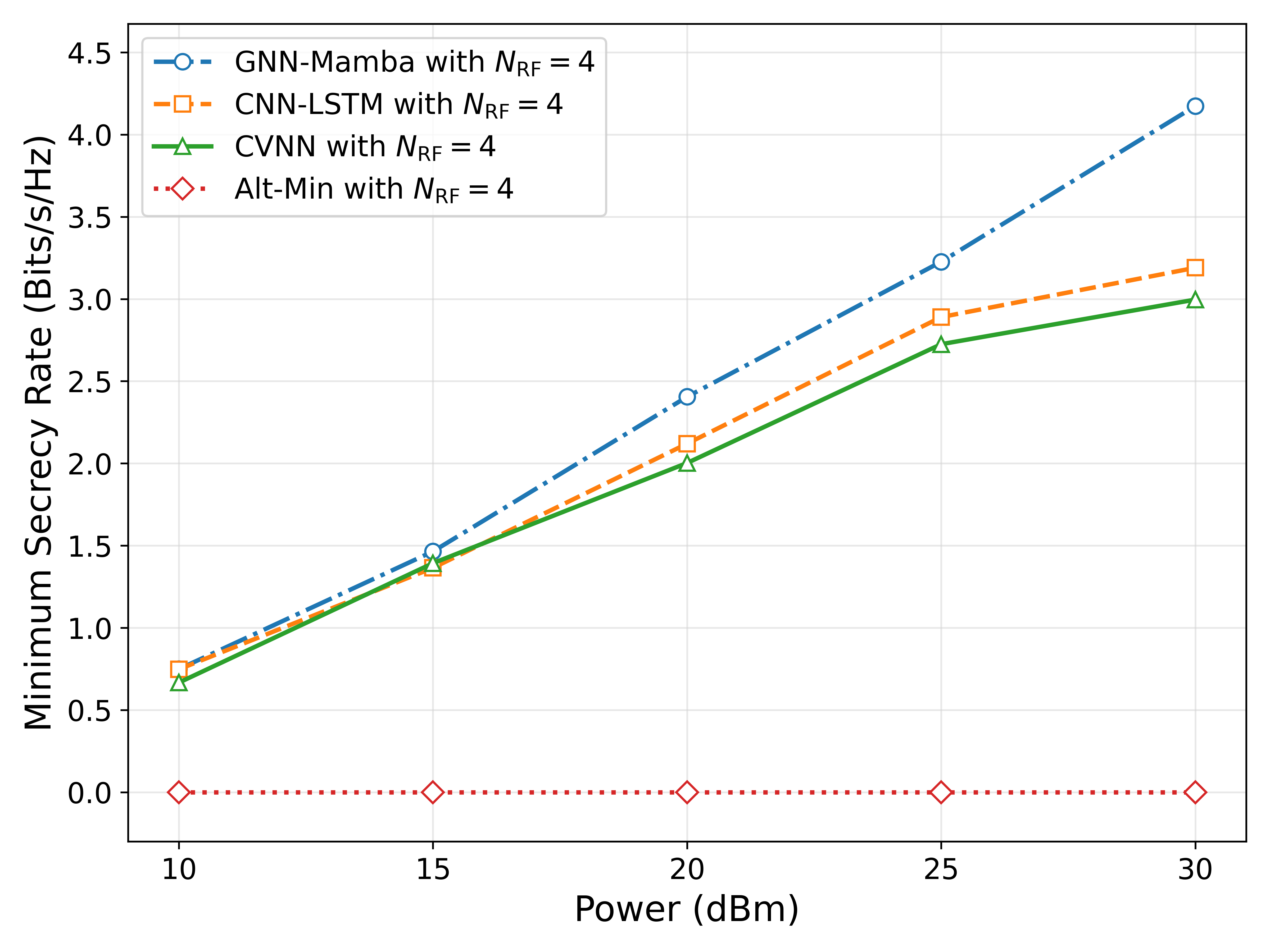} 
	\caption{ Minimum secrecy rate over the test samples versus transmit power for the proposed  scheme and different benchmarks with $L=3$, $K=3$, $M=3$, and $N_{\mathrm{RF}}=4$.}
	\label{fig:benchmark_learning}
\end{figure}
	\begin{figure}
	\vspace{-4mm}
	\centering
	\subfloat[\footnotesize{Secrecy rate versus $K$. }]
	{\includegraphics[width=1\linewidth]{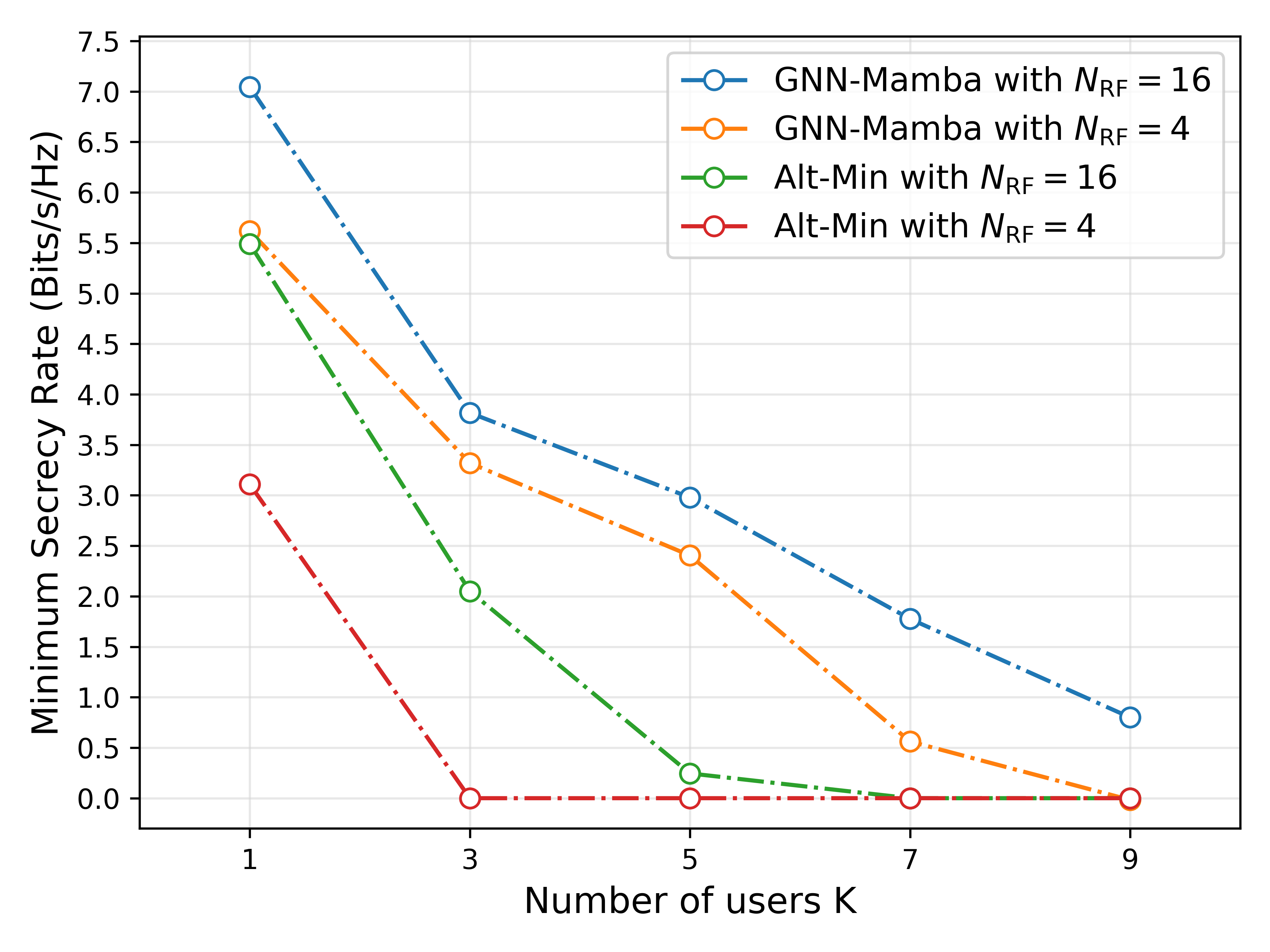}\label{Simul_5}}\\
	\subfloat[\footnotesize{Secrecy rate versus $M$.}]
	{\includegraphics[width=1\linewidth]{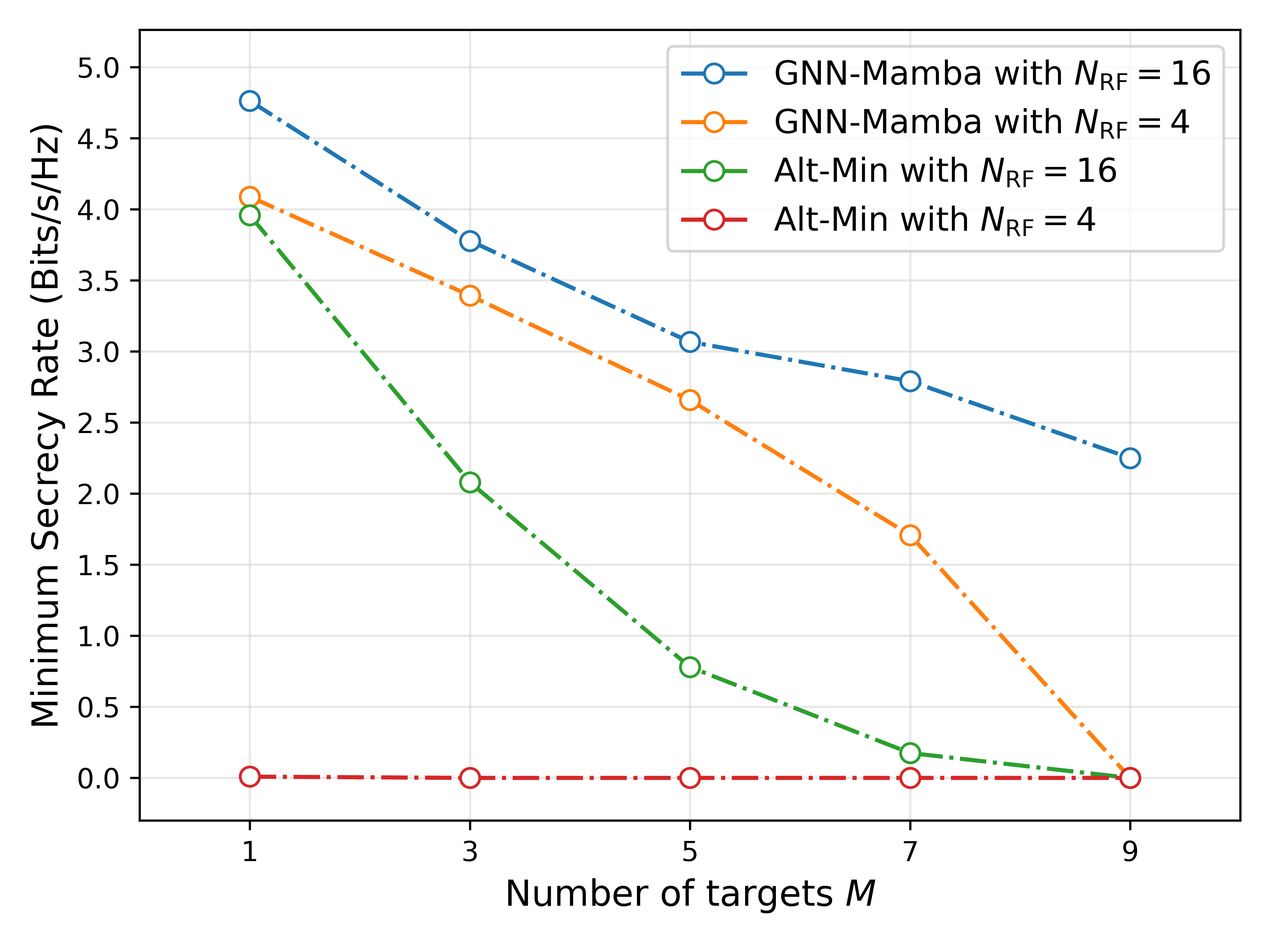}\label{Simul_6}}
	\caption{GNN vs. Alt-Min with varying $N_{\text{RF}}$, where $L{=}3$ and $P{=}25$~dBm.}
	\label{N_RF}
	\vspace{-4mm}
\end{figure}

\section{Simulation results}	
{In this section, numerical results are presented to validate the effectiveness of the proposed Mamba-empowered GNN framework for hybrid beamforming design in a cooperative THz-ISAC system. The simulation considers $N = 3$ BSs located at $(-53, -53, 0)$ m, $(53, -53, 0)$ m, and $(-53, 53, 0)$ m, respectively, with a radar receiver placed at the origin $(0, 0, -5)$ m. Users and targets are randomly distributed within a 10 m radius centered at $(0, 0, -10)$ m. The molecular absorption loss $k_{\text{abs}}(f_l)$ in the 100--450 GHz frequency band is modeled exploiting the HITRAN database under standard atmospheric conditions ($25^\circ$C, 1.5 g/m\textsuperscript{3} water vapor density)~\cite{b10}. 	
Additionally, assuming that the power budgets of three BSs are set the same as $P_n=P$, the CRB threshold $\Omega$ is set to $[0.025, 0.0084, 0.0038, 0.0012, 0.00038]$, corresponding to the power budget $P$ of $[10, 15, 20, 25, 30]$ dBm, respectively. The weighting parameters are set as $\omega_1 = -0.3$ and $\omega_2 = 0.1 $. The neural network is trained using the Adam optimizer with a learning rate of $1\times10^{-4}$ for 20 epochs, with each epoch consisting of 300 iterations.
Unless otherwise specified, the remaining simulation parameters are listed in Table~\ref{Setup}. { The conventional Alt-Min method~\cite{Alt-Min} is adopted as a benchmark. Specifically, Alt-Min first performs SVD on the sensing channel matrix $\mathbf{A}_l$ to obtain its null-space basis, based on which a fully digital ZF-based precoder $\mathbf{F}_{\mathrm{opt}}$ is constructed to suppress inter-user interference and enhance security. Then, under the considered TTD-assisted hybrid beamforming architecture, the analog beamformer, including the phase-shifter-based matrix $\mathbf{F}_{\mathrm{RF}}$ and the TTD matrix, together with the digital beamformer $\mathbf{F}_{\mathrm{BB}}$, is alternately optimized to approximate $\mathbf{F}_{\mathrm{opt}}$ in an iterative manner. For further comparison with representative learning-based designs, we include two deep learning baselines in the simulation study: the CVNN-based scheme from the conference version of this work~\cite{GLOBECOM25}, and a CNN-LSTM-based scheme adapted from~\cite{murshed2023cnn} and extended to the considered TTD-aided wideband hybrid architecture.}}

\begin{table}[!t]
	\caption{LIST OF SIMULATION PARAMETERS \label{Setup}}
	\centering
		\begin{tabular}{|p{5.5cm}|p{2.5cm}|}
			\hline
			\textbf{Parameters} & \textbf{Value}\\
			\hline
			Thz spectrum range, $f_\text{start} - f_\text{end}$ & $0.4 - 0.4385$ THz \\
			\hline
			Number of carriers, $L$ & 3 or 6  \\
			\hline
			Guard band bandwidth, $b_g$ & 1 GHz \cite{b1}   \\
			\hline
			Subcarrier Bandwidth, $b_s$ & 5 GHz \cite{b2}   \\
			\hline
			Transmit power of BS, $P$  & 10 $\sim$ 30 dBm   \\
			\hline	
			Noise power, $\sigma^2$ & $-92$ dBm \cite{b3}    \\
			\hline
			Number of antennas at the BS, $N_t$  & 128     \\
			\hline
			Number of RF chains at the BS, $N_{\text{RF}}$ & 4 or 16  \cite{b4}  \\
			\hline
			Number of antennas at the receiver, $N_r$ & 4     \\
			\hline
			Antenna element spacing &  375 $\mu m$     \\
			\hline
			Antenna gain $G_t$, $G_r$ & 20, 30 dBi  \cite{b5} \\
			\hline
			Number of OFDM symbols, $P$ & 256  \cite{b6} \\
			\hline
			Racian factor, $k_\text{H}$ & 20 dB \cite{b7}     \\
			\hline
			RCS of targets, $q$ & 1 \cite{Rician}     \\
			\hline
			Number of antennas per TTD, $J$ & 1 \cite{b8} \\
			\hline
			Maximum TTD delay, $\tau_{\max}$ & 20 ps \\
			\hline
		\end{tabular}
		\vspace{-5mm}
	\end{table}

	\begin{figure}
		\vspace{-1mm}
		\centering
		\includegraphics[width=1\linewidth]{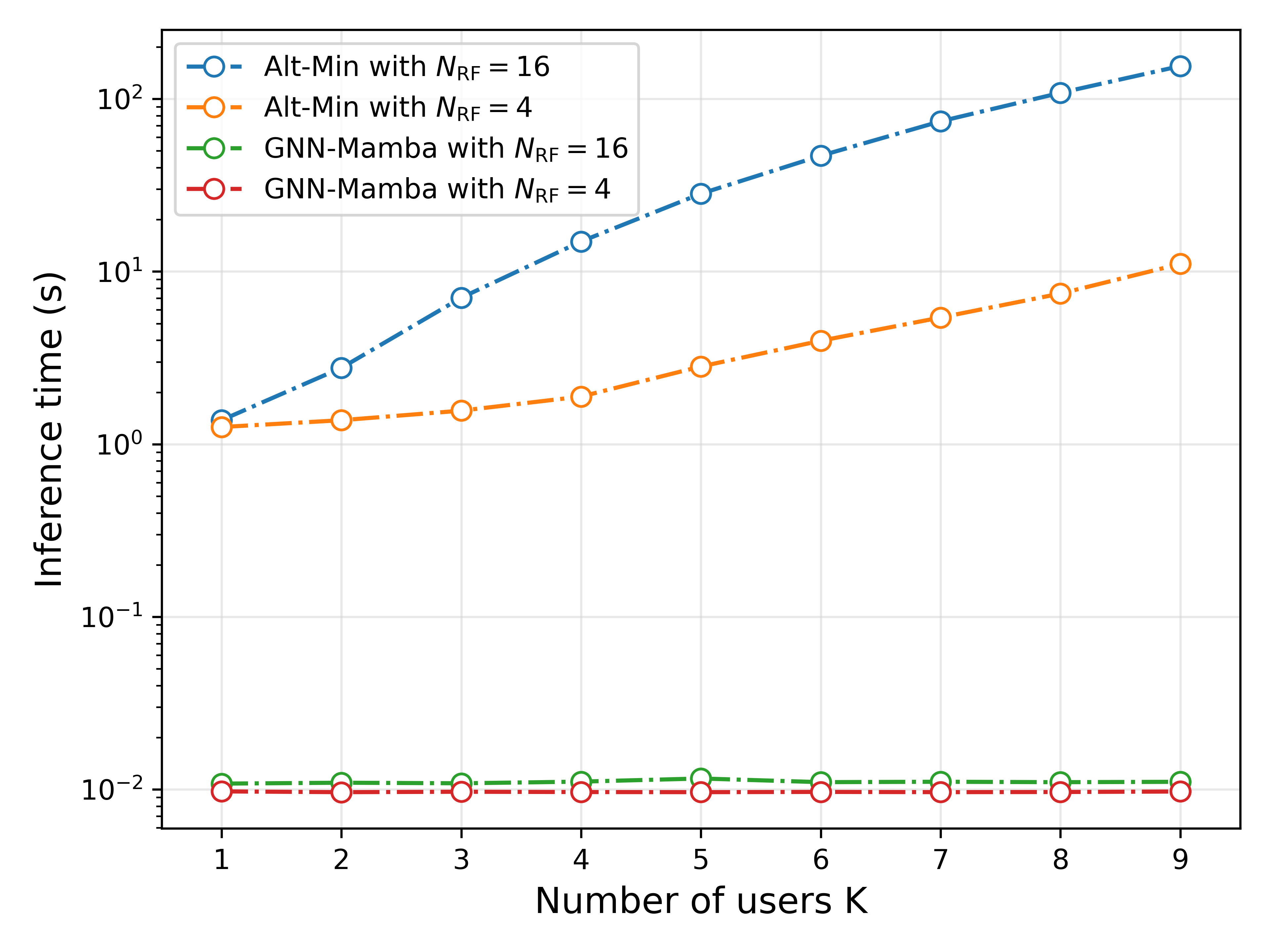}
		\caption{Beamforming computation time of GNN and Alt-Min with varying $K$, where $L=$3.}
		\vspace{-1mm}\label{Simul_7}
	\end{figure}
	
	\begin{figure*}
		\vspace{-2.7mm}
		\centering
		\includegraphics[width=1\linewidth]{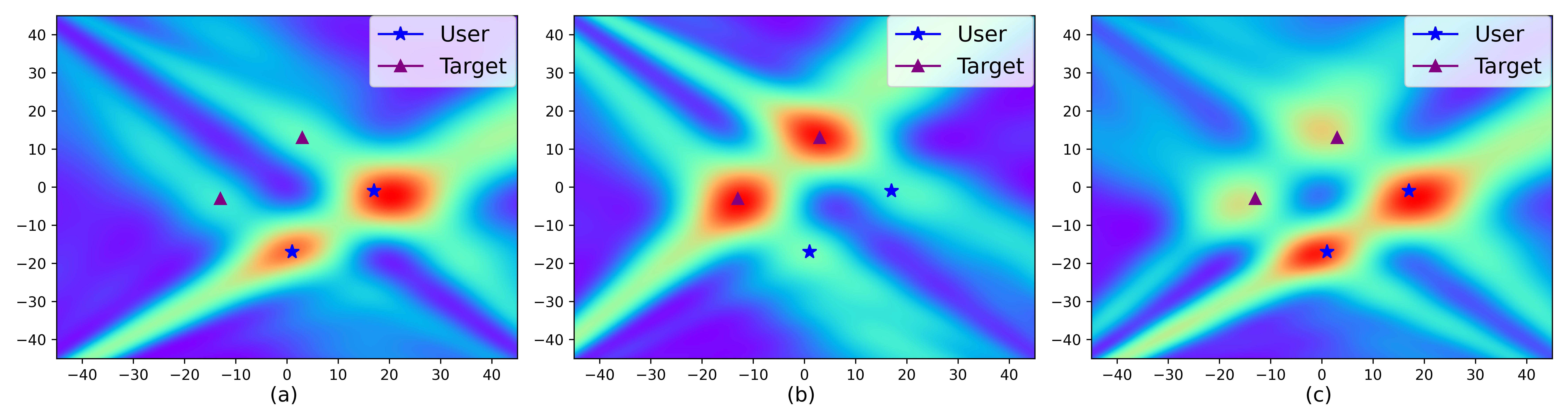}
		\caption{Beampatterns under the GNN-based analog beamforming scheme with $K{=}2$, $M{=}2$, $N_{\text{RF}}{=}4$.}
		\vspace{-5mm}\label{Simul_8}
	\end{figure*}
	
	{Figs.~\ref{fig:scalability_gnn}\subref{Simul_1} and \ref{fig:scalability_gnn}\subref{Simul_2} present the minimum secrecy rate with different transmit power budget, considering different subcarrier counts and RF chain configurations. From these results we draw several	key insights: 
	1) Impact of power $P$: The proposed method consistently outperforms the benchmark across all configurations and exhibits further improvement in the high-power regime, whereas the benchmark gradually saturates; 
	2) Impact of $L$: Increasing $L$ reduces the secrecy rate. With the fixed total power $P$, more subcarriers dilute each subcarrier signal power density. Poor channel conditions on any given subcarrier can then create a performance bottleneck for the	worst-case user, driving down the secrecy rate; 
	3) Impact of $N_{\text{RF}}$: When $N_{\text{RF}}$ is small, the Alt-Min method nearly collapses due to constrained beamforming degrees of freedom, tightly coupled optimization variables, and an inability to balance the sensing–communication power trade-off. In contrast, our Mamba-empowered GNN leverages graph-based message passing to decouple analog/digital precoders and the radar covariance matrix. Additionally, PAB is employed to budget the  communication and sensing power for balancing their performance.
	}
	

	%

	{Figs.~\ref{fig:scalability_gnn}\subref{Simul_3} and \ref{fig:scalability_gnn}\subref{Simul_4} show the minimum secrecy rate achieved by the model trained with $K=2$ and $M=3$ when tested under different numbers of sensing targets and communication users. The proposed \gls{GNN}-based architecture maintains favorable performance across these different configurations, indicating its generalization capability under topology variations. This behavior can be attributed to the permutation-invariant aggregation and structural inductive bias of \glspl{GNN}, under which the transmit design is generated through repeated aggregation of local interactions. As a result, the learned model captures transferable beamforming rules from reusable combinations of local interaction patterns, rather than from a fixed topology-specific mapping, which helps maintain its effectiveness under different network configurations.}

		{Fig.~\ref{fig:benchmark_learning} compares the minimum secrecy rate versus the transmit power budget for the proposed GNN-Mamba scheme and deep learning-based benchmarks. As shown in the figure, the proposed scheme consistently outperforms both benchmarks across the whole transmit-power range. Moreover, its performance advantage becomes increasingly evident as the transmit power budget increases, particularly in the medium- and high-power regimes. These results indicate that the proposed learning framework achieves more favorable secrecy-rate performance than the considered deep learning benchmarks, further reflecting the benefit of the Mamba-empowered GNN design together with the introduction of \gls{TTD}.}

	{Figs.~\ref{N_RF}\subref{Simul_5} and \ref{N_RF}\subref{Simul_6} show the minimum secrecy rate versus $M$ and $K$. The proposed method consistently outperforms the benchmark across all configurations. With four BS RF chains, our scheme maintains robust performance until $K + M$ approaches the number of RF chains. In contrast, the baseline degrades in small-scale configurations due to its inherent beamforming limitations with limited RF chains. These results demonstrate that our GNN-based approach effectively exploits \gls{ELAA} to design high-directional, spatially-isolated beams that simultaneously enhance sensing accuracy and communication security.
		}

	{Fig.~\ref{Simul_7} compares the per-sample beamforming computation times of the proposed method and the baseline. Our GNN-based approach achieves a consistently low inference time, scaling approximately linearly with $K$ due to its linear computational complexity. In contrast, the baseline suffers from high complexity and poor scalability, exhibiting a rapid increase in processing time with the number of users. These results confirm the fast inference capability of our approach.
	}

	{In Fig.~\ref{Simul_8}, we present the normalized beamforming gain on subcarrier~1 in the two-dimensional  space. In Fig.~\ref{Simul_8}{\color{red}(a)}, only communication security is considered, with secrecy-rate maximization as the design objective. In this case, the beamforming gain is mainly concentrated on the two legitimate users. In Fig.~\ref{Simul_8}{\color{red}(b)}, only sensing performance is considered, with the CRB as the optimization objective. Accordingly, the energy is primarily directed toward the two targets. In Fig.~\ref{Simul_8}{\color{red}(c)}, both communication security and sensing performance are jointly considered. In this case, part of the signal energy is allocated toward the targets to satisfy the sensing requirement, while the remaining energy is steered toward the users to maintain secure communication. These results highlight that different design objectives lead to distinct spatial beam patterns in the considered secure \gls{ISAC} system, while the joint design yields a trade-off between user service and target sensing.}
	\begin{figure}[t]
		\centering
		\includegraphics[width=0.9\linewidth]{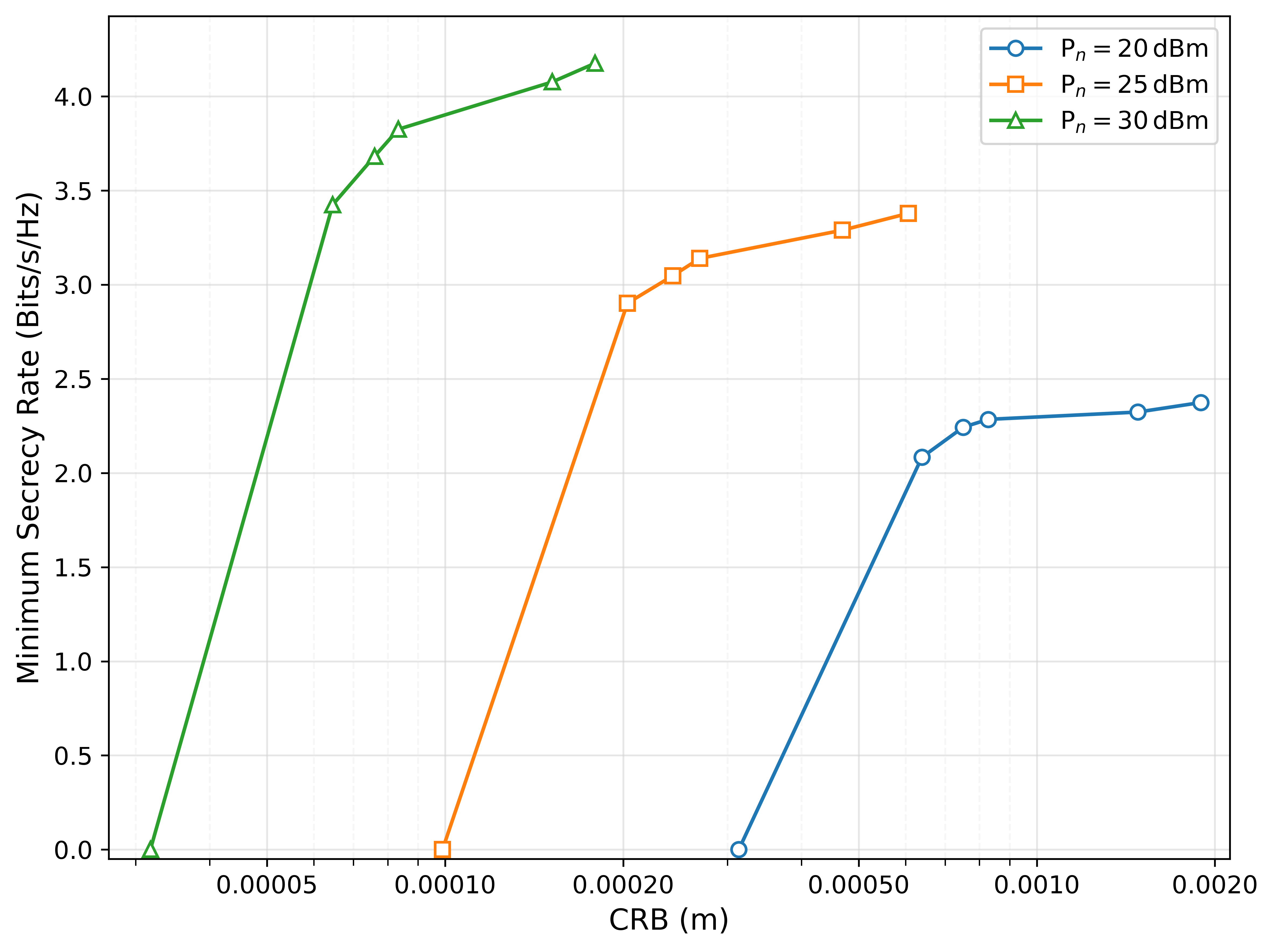}
		\caption{Minimum secrecy rate versus \gls{CRB} under different communication--sensing trade-off settings for $L=3$, $K=3$, $M=3$.}
		\label{fig:crb_secrecy_tradeoff}
	\end{figure}
	
	{
	Fig.~\ref{fig:crb_secrecy_tradeoff} illustrates the relationship between the  minimum secrecy rate and the  \gls{CRB} for different $P_n$. It can be observed that a higher secrecy rate generally corresponds to a larger CRB, whereas reducing the CRB enhances sensing accuracy at the cost of a degraded secrecy rate. This result clearly demonstrates the inherent tradeoff between communication security and sensing performance in the proposed design.
	}
	
	 {Fig.~\ref{fig:layers_robustness} shows the minimum secrecy rate versus the number of message-passing layers. It can be observed that the best performance is achieved with two layers, while the secrecy rate remains relatively stable as the depth is further increased within the tested range. This result suggests that the proposed framework is reasonably robust to the choice of network depth and does not exhibit severe performance degradation caused by oversmoothing under the considered setting. This behavior can be attributed to the fact that, at each layer, the Mamba block selectively propagates task-relevant features rather than uniformly mixing all neighboring information, while the residual connection preserves node-specific information from earlier layers. This helps preserve distinctions among node representations across layers. }

	\begin{figure}[t]
		\centering
		\includegraphics[width=3.4in]{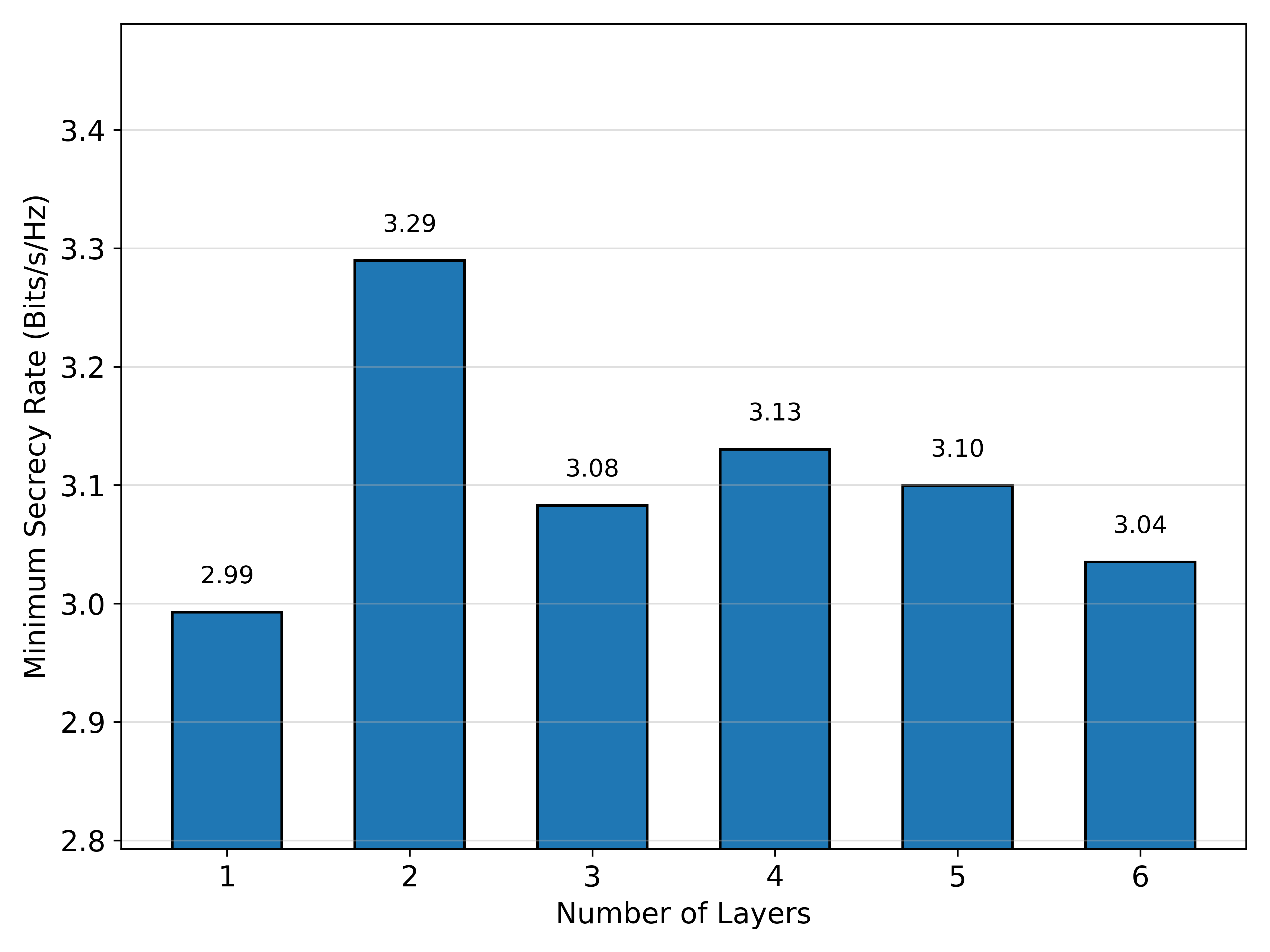}
		\caption{Minimum secrecy rate over the test samples versus the number of message-passing layers for $L=3$, $K=3$, $M=3$, and $P_n=25$ dBm.}
		\label{fig:layers_robustness}
	\end{figure}
	
	\section{Conclusions}
	This work investigated information security in cooperative bistatic ISAC networks, where multiple \gls{BS}s jointly serve downlink users while locating multiple targets. To secure confidential communications against wiretapping by malicious targets, we formulated a joint optimization problem for analog beamforming, \gls{TTD} networks, sensing signal, and digital precoders. Our objective was to maximize the minimum secrecy rate across users and subcarriers under constraints for target localization accuracy, defined by the \gls{CRB}. This challenging nonconvex problem, arising from coupled variables and the nonconvex \gls{CRB} constraint, was addressed through a novel Mamba-empowered \gls{GNN} framework. Our approach modeled the network interactions as a graph, where message passing generated optimized features for joint design, enhanced by a Mamba module for efficient feature exploitation. Simulation results demonstrated the superiority of the proposed method over both conventional optimization-based benchmarks and deep learning-based baselines. 
	{Several directions merit further investigation. Covert \gls{ISAC} represents a promising extension, where the presence or intent of dual-functional transmissions should be made less detectable by unauthorized wardens while preserving secure communication and sensing accuracy. This calls for joint beamforming and waveform designs that account for secrecy, sensing, and covertness requirements in a unified manner. Integrated communication--sensing--computation is also a relevant direction for cooperative \gls{ISAC} networks, where radio resources, sensing data processing, and edge computing resources should be jointly managed to support latency- and energy-constrained intelligent services.}
	
\appendix
{
	
\subsection{Proof of Lemma 1}
As the number of symbols $P$ is sufficiently large, we have
$
\underset{p\rightarrow+\infty}{\lim}\frac{1}{P}\sum_{p=1}^P\mathbf{c}_l(p)\mathbf{c}_l(p)^\mathrm{H}=	\begin{bmatrix}
	\mathbf{I}_{K} &  \mathbf{0} \\ 
	\mathbf{0}     &  \mathbf{S}_{l}  \\  
\end{bmatrix} 	.
$
Consequently, the covariance of $\mathrm{vec}(\mathbf{Y}_l^\mathrm{r})$ is determined by the noise component $\mathbf{n}_\mathrm{s}^l = \mathrm{vec}(\mathbf{N}_\mathrm{s}^l)$   \cite{b6}, i.e., 
$
\mathbf{n}_\mathrm{s}^l
= \left( 
\begin{bmatrix}
	\mathbf{c}_l^\mathrm{I}[p] \\
	\mathbf{c}_l^\mathrm{S}[p]
\end{bmatrix}^* \otimes \mathbf{I}_{N_\mathrm{r}} 
\right) \mathbf{z}_l^\mathrm{r}[p].
$
Since $\mathbf{z}_l^\mathrm{r}[p]\sim \mathcal{CN}(\mathbf{0}, \sigma_\mathrm{r}^2 \mathbf{I})$, we have	$
	\mathbb{E}(\mathbf{n}_\mathrm{s}^l) = \mathbf{0}. 
	$
	Besides,
	\begin{align}
		\boldsymbol{\Xi}_l
		=& \mathbb{E}\bigg\{ \frac{\sigma_{\mathrm{r}}^2}{P^2} \sum_{p=1}^{P} 
		\bigg( 
		\bigg[\begin{aligned}
			& (\mathbf{c}_l^\mathrm{I}[p])\\ 
			& (\mathbf{c}_l^\mathrm{S}[p])\\ 
		\end{aligned} \bigg]^*
		\otimes \mathbf{I}_{N_\mathrm{r}} \nonumber
		\bigg)
		\bigg( \bigg[\begin{aligned}
			& (\mathbf{c}_l^\mathrm{I}[p])\\ 
			& (\mathbf{c}_l^\mathrm{S}[p])\\ 
		\end{aligned} \bigg]^*
		\otimes \mathbf{I}_{N_\mathrm{r}} \bigg)^\mathrm{H}\bigg\}  \\
		=& \frac{\sigma_\mathrm{r}^2}{P} \, \text{blkdiag}(\mathbf{I}_{K}, \mathbf{S}^*_l) \otimes \mathbf{I}_{N_\mathrm{r}}.
	\end{align}

	\subsection{Proof of Theorem 1}
	Given the statistical characterization of the received signal presented in \textit{Lemma 1}, the \((i, j)\)-th entry of the FIM \( \mathcal{I}(\boldsymbol{\theta}, \boldsymbol{\theta}) \), for all \( i, j \in \{1, \ldots, 5M\} \), is typically given by
	\begin{align}
		\mathcal{I}(\boldsymbol{\theta}[i], \boldsymbol{\theta}[j]) 
		= 2 \Re \left\{ \left( \frac{\partial \boldsymbol{\mu}}{\partial \boldsymbol{\theta}[i]} \right)^\text{H} \boldsymbol{\Xi}^{-1} \frac{\partial \boldsymbol{\mu}}{\partial \boldsymbol{\theta}[j]} \right\}. \nonumber
	\end{align}
	By exploiting the block-diagonal structure of \( \boldsymbol{\Xi} \), we obtain
	\begin{align}
		\label{I(i,j)}
		&\mathcal{I}(\boldsymbol{\theta}[i], \boldsymbol{\theta}[j])\notag  \\ 
		&= \frac{2P}{\sigma_\mathrm{r}^2} \Re\Bigg\{\sum_{l=1}^{L} \text{tr}\Bigg( \frac{\partial(\mathbf{B}_l \mathbf{Q} \mathbf{A}_l\mathbf{F}_l \mathbf{D}_{l})}{\partial \boldsymbol{\theta}[i]}^\mathrm{H} \frac{\partial(\mathbf{B}_l \mathbf{Q} \mathbf{A}_l \mathbf{F}_l \mathbf{D}_{l})}{\partial \boldsymbol{\theta}[j]} \Bigg)\notag \\  &+ \sum_{g=1}^{NN_\mathrm{RF}} \sum_{\kappa=1}^{NN_\mathrm{RF}} \tilde{s}_{\kappa g,l} \frac{\partial(\mathbf{B}_l \mathbf{Q} \mathbf{A}_l \mathbf{F}_l \mathbf{s}_{\kappa,l})}{\partial \boldsymbol{\theta}[i]}^\mathrm{H} \frac{\partial(\mathbf{B}_l \mathbf{Q} \mathbf{A}_l \mathbf{F}_l \mathbf{s}_{g,l})}{\partial \boldsymbol{\theta}[j]} \Bigg\},
	\end{align}
	where $\tilde{s}_{\kappa g,l}$ represents the element in the $\kappa$-th row and $g$-th column of 
	$(\mathbf{S}^*_l)^{-1}$. Building on the general expression in \eqref{I(i,j)} and the block structure presented in \textit{Theorem~1}, each $M \times M$ subblock of the FIM can be further derived into a closed-form expression exploiting the explicit form of the corresponding partial derivatives \cite{b6}.
	To proceed, we examine a typical entry $\mathcal{I}_{\mathbf{u}[i] \mathbf{v}[j]}$ with $\mathbf{u}, \mathbf{v} \in \{ \mathbf{x}^\text{t}, \mathbf{y}^\text{t}, \mathbf{z}^\text{t} \}$ and $1 \leq i, j \leq M$.
	An explicit expression is obtained by expanding the term $ \text{tr}\Big( \frac{\partial(\mathbf{B}_l \mathbf{Q} \mathbf{A}_l\mathbf{F}_l \mathbf{D}_{l})}{\partial \mathbf{u}[i]}^\mathrm{H} \frac{\partial(\mathbf{B}_l \mathbf{Q} \mathbf{A}_l \mathbf{F}_l \mathbf{D}_{l})}{\partial  \mathbf{v}[j]} \Big)$  as follows:
	\begin{align}
		\label{first_term}
		&\text{tr}\Big\{\bigg(\frac{\partial(\mathbf{B}_l \mathbf{Q} \mathbf{A}_l \mathbf{F}_l\mathbf{D}_{l})}{\partial  \mathbf{u}[i]}\bigg)^\mathrm{H}  \frac{\partial 	(\mathbf{B}_l \mathbf{Q} \mathbf{A}_l \mathbf{F}_l\mathbf{D}_{l})}{\partial \mathbf{v}[j]} \Big\} \nonumber\\ 
		&=
		\bigg\{\mathbf{e}_i^\mathrm{T} (\dot{\mathbf{B}_l}_u)^\mathrm{H} \dot{\mathbf{B}_l}_v  \mathbf{e}_j \mathbf{e}_j^\mathrm{T} \mathbf{Q} \mathbf{A}_l \mathbf{F}_l\mathbf{D}_l \mathbf{D}_l^\mathrm{H} \mathbf{F}_l^\mathrm{H} (\mathbf{A}_l)^\mathrm{H}\mathbf{Q}^\mathrm{H} \mathbf{e}_i\nonumber\\
		&+ 
		\mathbf{e}_i^\mathrm{T} (\dot{\mathbf{B}_l}_u)^\mathrm{H} \mathbf{B}_l   \mathbf{e}_j \mathbf{e}_j^\mathrm{T} \mathbf{Q} (\dot{\mathbf{A}_l}_v) \mathbf{F}_l\mathbf{D}_l \mathbf{D}_l^\mathrm{H} \mathbf{F}_l^\mathrm{H} (\mathbf{A}_l)^\mathrm{H} \mathbf{Q}^\mathrm{H} \mathbf{e}_i  \nonumber\\
		&+ 
		\mathbf{e}_i^\mathrm{T}  (\mathbf{B}_l)^\mathrm{H} \dot{\mathbf{B}_l}_v \mathbf{e}_j \mathbf{e}_j^\mathrm{T} \mathbf{Q} \mathbf{A}_l \mathbf{F}_l\mathbf{D}_l \mathbf{D}_l^\mathrm{H} \mathbf{F}_l^\mathrm{H} (\dot{\mathbf{A}_l}_u)^\mathrm{H}\mathbf{Q}^\mathrm{H} \mathbf{e}_i\nonumber\\
		&+ 
		\mathbf{e}_i^\mathrm{T}  (\mathbf{B}_l)^\mathrm{H} \mathbf{B}_l  \mathbf{e}_j \mathbf{e}_j^\mathrm{T} \mathbf{Q} (\dot{\mathbf{A}_l}_v) \mathbf{F}_l\mathbf{D}_l \mathbf{D}_l^\mathrm{H}\mathbf{F}_l^\mathrm{H}(\dot{\mathbf{A}_l}_u)^\mathrm{H} \mathbf{Q}^\mathrm{H}\mathbf{e}_i \bigg\},
	\end{align}
	where \( \mathbf{e}_i \) and \( \mathbf{e}_j \) denote the \( i \)-th and \( j \)-th columns of the identity matrix, respectively.
	Expanding the term $\big(\frac{\partial 	(\mathbf{B}_l \mathbf{Q} \mathbf{A}_l \mathbf{F}_l \mathbf{s}_{\kappa,l})}{\partial \mathbf{u}[i]}\big)^\mathrm{H} 	\frac{\partial 	(\mathbf{B}_l \mathbf{Q} \mathbf{A}_l \mathbf{F}_l \mathbf{s}_{g,l})}{\partial \mathbf{v}[j]}$  yields:
	\begin{align}	
		\label{second_term}
		&\Bigg\{ \bigg(\frac{\partial 	(\mathbf{B}_l \mathbf{Q} \mathbf{A}_l \mathbf{F}_l \mathbf{s}_{\kappa,l})}{\partial \mathbf{u}[i]}\bigg)^\mathrm{H} 	\frac{\partial 	(\mathbf{B}_l \mathbf{Q} \mathbf{A}_l \mathbf{F}_l \mathbf{s}_{g,l})}{\partial \mathbf{v}[j]}\Bigg\} \notag \\ 
		&=
		\Big\{
		\mathbf{e}_j^\mathrm{T} \mathbf{Q} \mathbf{A}_l \mathbf{F}_l\mathbf{s}_{g,l} 	\mathbf{s}_{\kappa,l}^\mathrm{H}\mathbf{F}_l^\mathrm{H} (\mathbf{A}_l)^\mathrm{H} \mathbf{Q}^\mathrm{H} \mathbf{e}_i \mathbf{e}_i^\mathrm{T} (\dot{\mathbf{B}_l}_u)^\mathrm{H} \dot{\mathbf{B}_l}_v \mathbf{e}_j\notag \\ 
		&+ 
		\mathbf{e}_j^\mathrm{T} \mathbf{Q}(\dot{\mathbf{A}_l}_v) \mathbf{F}_l\mathbf{s}_{g,l} 	\mathbf{s}_{\kappa,l}^\mathrm{H} \mathbf{F}_l^\mathrm{H} (\mathbf{A}_l)^\mathrm{H} \mathbf{Q}^\mathrm{H} \mathbf{e}_i \mathbf{e}_i^\mathrm{T} (\dot{\mathbf{B}_l}_u)^\mathrm{H} 	\mathbf{B}_l  \mathbf{e}_j\notag \\ 
		&+ 
		\mathbf{e}_j^\mathrm{T} \mathbf{Q} \mathbf{A}_l \mathbf{F}_l\mathbf{s}_{g,l} 	\mathbf{s}_{\kappa,l}^\mathrm{H} \mathbf{F}_l^\mathrm{H} (\dot{\mathbf{A}_l}_u)^\mathrm{H} \mathbf{e}_i \mathbf{e}_i^\mathrm{T} \mathbf{Q}^\mathrm{H} (\mathbf{B}_l)^\mathrm{H}\dot{\mathbf{B}_l}_v   \mathbf{e}_j\notag \\ 
		&+ 
		\mathbf{e}_j^\mathrm{T} \mathbf{Q} (\dot{\mathbf{A}_l}_v) \mathbf{F}_l \mathbf{s}_{g,l} 	\mathbf{s}_{\kappa,l}^\mathrm{H} \mathbf{F}_l^\mathrm{H} (\dot{\mathbf{A}_l}_u)^\mathrm{H} \mathbf{e}_i \mathbf{e}_i^\mathrm{T} \mathbf{Q}^\mathrm{H} (\mathbf{B}_l)^\mathrm{H} 	\mathbf{B}_l\mathbf{e}_j
		\Big\}.
	\end{align}
	Accordingly, the expression for $\mathcal{I}_{\mathbf{u}[i] \mathbf{v}[j]}$ is given by:
	\begin{align}
		&\mathcal{I}_{\mathbf{u}[i] \mathbf{v}[j]} \nonumber\\
		=&\frac{P}{\sigma_\mathrm{r}^2}\Re\Bigg\{\sum_{l=1}^{L} \nonumber
		\mathbf{e}_i^\mathrm{T} (\dot{\mathbf{B}_l}_u)^\mathrm{H} \dot{\mathbf{B}_l}_v  \mathbf{e}_j \mathbf{e}_j^\mathrm{T} \mathbf{Q} \mathbf{A}_l \mathbf{F}_l\mathbf{D}_l \mathbf{D}_l^\mathrm{H} \mathbf{F}_l^\mathrm{H} (\mathbf{A}_l)^\mathrm{H}\mathbf{Q}^\mathrm{H} \mathbf{e}_i\\\nonumber
		+& 
		\mathbf{e}_i^\mathrm{T} (\dot{\mathbf{B}_l}_u)^\mathrm{H} \mathbf{B}_l  \mathbf{e}_j \mathbf{e}_j^\mathrm{T} \mathbf{Q} (\dot{\mathbf{A}_l}_v) \mathbf{F}_l\mathbf{D}_l \mathbf{D}_l^\mathrm{H} \mathbf{F}_l^\mathrm{H} (\mathbf{A}_l)^\mathrm{H} \mathbf{Q}^\mathrm{H} \mathbf{e}_i  \\\nonumber
		+& 
		\mathbf{e}_i^\mathrm{T}  (\mathbf{B}_l)^\mathrm{H} \dot{\mathbf{B}_l}_v \mathbf{e}_j \mathbf{e}_j^\mathrm{T} \mathbf{Q} \mathbf{A}_l \mathbf{F}_l\mathbf{D}_l \mathbf{D}_l^\mathrm{H} \mathbf{F}_l^\mathrm{H} (\dot{\mathbf{A}_l}_u)^\mathrm{H} \mathbf{Q}^\mathrm{H} \mathbf{e}_i\\\nonumber
		+ &
		\mathbf{e}_i^\mathrm{T}  (\mathbf{B}_l)^\mathrm{H} \mathbf{B}_l  \mathbf{e}_j \mathbf{e}_j^\mathrm{T} \mathbf{Q} (\dot{\mathbf{A}_l}_v) \mathbf{F}_l\mathbf{D}_l \mathbf{D}_l^\mathrm{H}\mathbf{F}_l^\mathrm{H}(\dot{\mathbf{A}_l}_u)^\mathrm{H} \mathbf{Q}^\mathrm{H} \mathbf{e}_i  \notag \\  \notag 
		+& \sum_{l=1}^{L} \sum_{g=1}^{NN_\mathrm{RF}}\nonumber \sum_{\kappa=1}^{NN_\mathrm{RF}}\\\nonumber
		&\tilde{s}_{\kappa g,l} \Big(\mathbf{e}_j^\mathrm{T} \mathbf{Q} \mathbf{A}_l \mathbf{F}_l\mathbf{s}_{g,l} 	\mathbf{s}_{\kappa,l}^\mathrm{H}\mathbf{F}_l^\mathrm{H} (\mathbf{A}_l)^\mathrm{H} \mathbf{Q}^\mathrm{H} \mathbf{e}_i \mathbf{e}_i^\mathrm{T} (\dot{\mathbf{B}_l}_u)^\mathrm{H} \dot{\mathbf{B}_l}_v \mathbf{e}_j\\\nonumber
		+ &
		\mathbf{e}_j^\mathrm{T} \mathbf{Q}(\dot{\mathbf{A}_l}_v) \mathbf{F}_l\mathbf{s}_{g,l} 	\mathbf{s}_{\kappa,l}^\mathrm{H} \mathbf{F}_l^\mathrm{H} (\mathbf{A}_l)^\mathrm{H} \mathbf{Q}^\mathrm{H} \mathbf{e}_i \mathbf{e}_i^\mathrm{T} (\dot{\mathbf{B}_l}_u)^\mathrm{H} 	\mathbf{B}_l  \mathbf{e}_j  \notag \\
		+&
		\mathbf{e}_j^\mathrm{T} \mathbf{Q} \mathbf{A}_l \mathbf{F}_l\mathbf{s}_{g,l} \nonumber
		\mathbf{s}_{\kappa,l}^\mathrm{H} \mathbf{F}_l^\mathrm{H} (\dot{\mathbf{A}_l}_u)^\mathrm{H}\mathbf{Q}^\mathrm{H} \mathbf{e}_i \mathbf{e}_i^\mathrm{T}  (\mathbf{B}_l)^\mathrm{H}\dot{\mathbf{B}_l}_v  \mathbf{e}_j\\
		+&
		\mathbf{e}_j^\mathrm{T} \mathbf{Q} (\dot{\mathbf{A}_l}_v) \mathbf{F}_l \mathbf{s}_{g,l} 	\mathbf{s}_{\kappa,l}^\mathrm{H} \mathbf{F}_l^\mathrm{H} (\dot{\mathbf{A}_l}_u)^\mathrm{H}\mathbf{Q}^\mathrm{H} \mathbf{e}_i \mathbf{e}_i^\mathrm{T}  (\mathbf{B}_l)^\mathrm{H} 	\mathbf{B}_l\mathbf{e}_j \Big) \Bigg\}. 
	\end{align}
	The resulting subblock $\mathcal{I}_{\mathbf{u} \mathbf{v}}$ is given by
	\begin{align}
		&\mathcal{I}_{\mathbf{u} \mathbf{v}} \notag \\
		=&\frac{P}{\sigma_\mathrm{r}^2}\Bigg\{\sum_{l=1}^{L}
		\bigg((\dot{\mathbf{B}_l}_\mathbf{u})^\mathrm{H} \dot{\mathbf{B}_l}_\mathbf{v} \bigg)  \odot \bigg(\mathbf{Q} \mathbf{A}_l \mathbf{F}_l\mathbf{D}_l \mathbf{D}_l^\mathrm{H} \mathbf{F}_l^\mathrm{H} (\mathbf{A}_l)^\mathrm{H}\mathbf{Q}^\mathrm{H}  \bigg)^{\mathrm{T}}\notag \\
		+ &
		\bigg( (\dot{\mathbf{B}_l}_\mathbf{u})^\mathrm{H} \mathbf{B}_l \bigg)    \odot \bigg( \mathbf{Q} (\dot{\mathbf{A}_l}_\mathbf{v}) \mathbf{F}_l\mathbf{D}_l \mathbf{D}_l^\mathrm{H} \mathbf{F}_l^\mathrm{H} (\mathbf{A}_l)^\mathrm{H} \mathbf{Q}^\mathrm{H}\bigg)^{\mathrm{T}}  \notag \\ 
		+ &\nonumber
		\bigg( (\mathbf{B}_l)^\mathrm{H} \dot{\mathbf{B}_l}_\mathbf{v}  \bigg) \odot \bigg(   \mathbf{Q} \mathbf{A}_l \mathbf{F}_l\mathbf{D}_l \mathbf{D}_l^\mathrm{H} \mathbf{F}_l^\mathrm{H} (\dot{\mathbf{A}_l}_\mathbf{u})^\mathrm{H} \mathbf{Q}^\mathrm{H}  \bigg)^{\mathrm{T}} \notag \\
		+ &
		\bigg( (\mathbf{B}_l)^\mathrm{H} \mathbf{B}_l  \bigg) \odot  \bigg( \mathbf{Q} (\dot{\mathbf{A}_l}_\mathbf{v}) \mathbf{F}_l\mathbf{D}_l \mathbf{D}_l^\mathrm{H}\mathbf{F}_l^\mathrm{H}(\dot{\mathbf{A}_l}_\mathbf{u})^\mathrm{H} \mathbf{Q}^\mathrm{H}\bigg)^{\mathrm{T}}   \notag \\ 
		+&  \sum_{l=1}^{L} \sum_{g=1}^{NN_\mathrm{RF}} \sum_{\kappa=1}^{NN_\mathrm{RF}}\notag \\
		&\tilde{s}_{\kappa g,l}  \bigg( \bigg( \mathbf{Q} \mathbf{A}_l \mathbf{F}_l\mathbf{s}_{g,l} 	\mathbf{s}_{\kappa,l}^\mathrm{H}\mathbf{F}_l^\mathrm{H} (\mathbf{A}_l)^\mathrm{H} \mathbf{Q}^\mathrm{H}\bigg)^\mathrm{T} \odot \bigg( (\dot{\mathbf{B}_l}_\mathbf{u})^\mathrm{H} \dot{\mathbf{B}_l}_\mathbf{v} \bigg)\notag \\ 
		+& 
		\bigg( \mathbf{Q}(\dot{\mathbf{A}_l}_\mathbf{v}) \mathbf{F}_l\mathbf{s}_{g,l} 	\mathbf{s}_{\kappa,l}^\mathrm{H} \mathbf{F}_l^\mathrm{H} (\mathbf{A}_l)^\mathrm{H} \mathbf{Q}^\mathrm{H} \bigg)^\mathrm{T}  \odot \bigg( (\dot{\mathbf{B}_l}_\mathbf{u})^\mathrm{H} 	\mathbf{B}_l \bigg) \notag \\
		+& \notag
		\bigg(  \mathbf{Q} \mathbf{A}_l \mathbf{F}_l\mathbf{s}_{g,l} 	\mathbf{s}_{\kappa,l}^\mathrm{H} \mathbf{F}_l^\mathrm{H} (\dot{\mathbf{A}_l}_\mathbf{u})^\mathrm{H}\mathbf{Q}^\mathrm{H} \bigg)^\mathrm{T}  \odot \bigg(  (\mathbf{B}_l)^\mathrm{H}\dot{\mathbf{B}_l}_\mathbf{v}   \bigg )\notag \\ 
		+&
		\bigg(  \mathbf{Q} (\dot{\mathbf{A}_l}_\mathbf{v}) \mathbf{F}_l \mathbf{s}_{g,l} 	\mathbf{s}_{\kappa,l}^\mathrm{H} \mathbf{F}_l^\mathrm{H}  (\dot{\mathbf{A}_l}_\mathbf{u})^\mathrm{H}\mathbf{Q}^\mathrm{H} \bigg)^\mathrm{T}   \odot \bigg(  (\mathbf{B}_l)^\mathrm{H} 	\mathbf{B}_l\bigg)  \bigg) \Bigg\}.  
	\end{align}
	Following a similar derivation, the terms \( \mathcal{I}_{\mathbf{u} \mathbf{q}} \) and  \( \mathcal{I}_{\mathbf{q} \mathbf{q}} \) can be obtained as follows:
	\begin{align}
		\label{fuq}
		&\mathcal{I}_{\mathbf{u}\mathbf{q}} \notag \\
		=&\frac{P}{\sigma_\mathrm{r}^2}\Bigg\{\sum_{l=1}^{L} 
		\bigg(
		(\dot{\mathbf{B}_l}_u)^\mathrm{H} \mathbf{B}_l\bigg) \odot\bigg(
		\mathbf{A}_l \mathbf{F}_l \mathbf{D}_{l}  \mathbf{D}_l^\mathrm{H} \mathbf{F}_l^\mathrm{H}(\mathbf{A}_l)^\mathrm{H} \mathbf{Q}^\mathrm{H} \bigg)^{\mathrm{T}}\notag \\
		+& 
		\bigg(	 (\mathbf{B}_l)^\mathrm{H} \mathbf{B}_l\bigg)\odot 	\bigg(\mathbf{A}_l \mathbf{F}_l\mathbf{D}_{l} \mathbf{D}_l^\mathrm{H}  \mathbf{F}_l^\mathrm{H}  (\dot{\mathbf{A}_l}_u)^\mathrm{H} \mathbf{Q}^\mathrm{H}\bigg)^{\mathrm{T}} \notag \\
		+&\sum_{l=1}^{L} \sum_{g=1}^{NN_\mathrm{RF}} \sum_{\kappa=1}^{NN_\mathrm{RF}} \notag \\
		&\tilde{s}_{\kappa g,l} 	\bigg(
		\bigg( \mathbf{A}_l\mathbf{F}_l\mathbf{s}_{g,l} 
		\mathbf{s}_{\kappa,l}^{\mathrm{H}} \mathbf{F}_l^{\mathrm{H}} \left(\mathbf{A}_l\right)^{\mathrm{H}} \mathbf{Q}^{\mathrm{H}} \bigg)^\mathrm{T}  
		\odot \bigg(  \left(\dot{\mathbf{B}_l}_u\right)^{\mathrm{H}}\mathbf{B}_l\bigg) \notag \\
		+ & \bigg( \mathbf{A}_l\mathbf{F}_l\mathbf{s}_{g,l}  \mathbf{s}_{\kappa,l}^{\mathrm{H}} \mathbf{F}_l^{\mathrm{H}} \left(\dot{\mathbf{A}_l}_u\right)^{\mathrm{H}}\mathbf{Q}^{\mathrm{H}} \bigg)^\mathrm{T}  \odot \bigg(   \left(\mathbf{B}_l\right)^{\mathrm{H}}\mathbf{B}_l\bigg) \bigg)\Bigg\},
	\end{align}
	
	\begin{align}
		\label{fqq}
		\mathcal{I}_{\mathbf{q}\mathbf{q}}=&\frac{P}{\sigma_\mathrm{r}^2} \Bigg\{\sum_{i=1}^{L} 
		\bigg( (\mathbf{B}_l)^\mathrm{H} \mathbf{B}_l \bigg) \odot  \bigg(\mathbf{A}_l \mathbf{F}_l\mathbf{D}_{l}\mathbf{D}_l^\mathrm{H} \mathbf{F}_l^\mathrm{H}(\mathbf{A}_l)^\mathrm{H}\bigg)^\mathrm{T} \notag \\
		+& \sum_{i=1}^{L} \sum_{g=1}^{NN_\mathrm{RF}} \sum_{\kappa=1}^{NN_\mathrm{RF}}\tilde{s}_{\kappa g,l} 	\bigg( \mathbf{A}_l \mathbf{F}_l\mathbf{s}_{g,l} \mathbf{s}_{\kappa,l}^\mathrm{H}  \mathbf{F}_l^\mathrm{H}  (\mathbf{A}_l)^\mathrm{H} \bigg)^\mathrm{T}  \notag \\
		\odot& \bigg(   (\mathbf{B}_l)^\mathrm{H} \mathbf{B}_l  \bigg) \Bigg\}.
	\end{align}
	}

	
	\ifCLASSOPTIONcaptionsoff
	\newpage
	\fi
	\bibliographystyle{IEEEtran}
	\bibliography{reference}
\end{document}